\documentclass{article}
\usepackage{arxiv}

\usepackage[utf8]{inputenc}
\usepackage[T1]{fontenc}
\usepackage{microtype}

\usepackage{amsmath, amssymb, amsthm, amsfonts} 

\usepackage{booktabs}     
\usepackage{float}        

\usepackage{graphicx}     
\usepackage{placeins}     
\usepackage{enumitem}   

\usepackage{hyperref}

\makeatletter
\newenvironment{landscapeblock}{%
  \clearpage
  \newgeometry{landscape,left=20mm,right=20mm,top=15mm,bottom=15mm,headheight=15pt,headsep=12pt}%
  \let\ps@plain\ps@fancy
  \pagestyle{fancy}\thispagestyle{fancy}%
}{%
  \clearpage
  \restoregeometry
  \pagestyle{fancy}%
}
\makeatother

\newtheorem{theorem}{Theorem}

\newtheorem{corollary}[theorem]{Corollary}

\newtheorem{definition}[theorem]{Definition}
\newtheorem{example}[theorem]{Example}

\newtheorem{lemma}[theorem]{Lemma}

\newtheorem{proposition}[theorem]{Proposition}
\newtheorem{remark}[theorem]{Remark}

\date{\today}

\author{Jos\'e M. Amig\'o\thanks{Corresponding author.} \\
Centro de Investigaci\'{o}n Operativa, \\
Universidad Miguel Hern\'{a}ndez, \\
03202 Elche, Spain \\
\texttt{jm.amigo@umh.es}
\And
Roberto Dale \\
Centro de Investigaci\'{o}n Operativa, \\
Universidad Miguel Hern\'{a}ndez, \\
03202 Elche, Spain \\
\texttt{rdale@umh.es}
}

\begin{document}

\title{\boldmath Permutation-Based Distances for Groups and Group-Valued Time Series}

\maketitle

\begin{abstract}
Permutations on a set, endowed with function composition,
build a group called a symmetric group. In addition to their algebraic
structure, symmetric groups have two metrics that are of particular interest to us here: the Cayley distance and the Kendall tau distance. In fact, the aim of this paper is to introduce the concept of distance in a general finite group based on them. The main tool that we use to this end is Cayley's theorem, which states that any finite group is isomorphic to a subgroup of a certain symmetric group. We also discuss the advantages and disadvantage of these permutation-based distances compared to the conventional generator-based distances in finite groups. The reason why we
are interested in distances on groups is that finite groups appear in
symbolic representations of time series, most notably in the so-called
ordinal representations, whose symbols are precisely permutations, usually called ordinal patterns in that context. The natural extension from groups to group-valued time series is also discussed, as well as how such metric tools can be applied in time series analysis. Both theory and applications are illustrated with examples and numerical simulations.
\end{abstract}


\section{Introduction}
\label{sec:Intro}

Symbolic representation of real-valued times series is a usual and useful
tool in data analysis, where numbers are replaced by discrete
``symbols'', {in order to gain more tools and insights~\cite{Hirata2023}.} So to speak, symbolic representations coarse-grain the data in such a way that the information retained is sufficient for the purposes of the analysis. From a mathematical point of view, this technique consists of partitioning the state space, both in
statistics and nonlinear methods. Traditional examples include binning and thresholding. More recently, Bandt and Pompe~\cite{Bandt2002} proposed to use ordinal patterns, which are the rank vectors of sliding windows along a time series, the size of the windows being the length of the ordinal patterns. Since then, ordinal representations, i.e., symbolic representations with ordinal patterns, have become a popular technique among
data analysts. {Common applications of ordinal patterns include classification using ordinal pattern-based indices~\mbox{\cite{Keller2003,Graff2013,Schlemmer2018}}, discrimination of chaotic signals from white noise~\cite{Rosso2007,Amigo2007}, characterization of dynamics and couplings~\cite{Bandt2002,Monetti2009,Parlitz2013} and nonparametric tests
of serial dependence~\cite{Weiss2022,Weiss2022B}, to mention a few.} For general
overviews, see~\cite{Amigo2013,Leyva2022,Amigo2023}.

More importantly for the topic of this paper, ordinal patterns of any given
length $L\geq 2$ can be interpreted as permutations (i.e., bijections) on
any set of $L$ elements, say, $\{1,2,\ldots,L\}$. In fact, the Shannon entropy
of a probability distribution of ordinal patterns is called permutation
entropy~\cite{Bandt2002}, and the same happens with any other entropic
functional based on ordinal pattern probability distributions, e.g.,
divergence, mutual information, or statistical complexity. A potential
advantage of viewing ordinal patterns of length $L$ as permutations is that
the latter build a group, namely, the symmetric group of degree $L$, denoted
by $\mathrm{Sym}(L)$, where the binary operation is function composition. In
fact, the algebraic structure of $\mathrm{Sym}(L)$ provides additional
leverage to ordinal representations that can be harnessed in time series
analysis. An example of this is the concept of transcript introduced in \cite%
{Monetti2009}.

More generally, symbolic representations whose symbols are elements of a
group are called algebraic representations, an ordinal representation being
an algebraic representation with alphabet $\mathrm{Sym}(L)$. Actually, most
results for ordinal representations can be readily generalized to algebraic
representations whose alphabets are any other finite group $\mathcal{G}$.
This is not surprising if no particular property of $\mathrm{Sym}(L)$ is
used in a given proof or application. There may be another, more theoretical
reason for this. According to Cayley's theorem~\cite{Herstein1996}, any
finite group $\mathcal{G}$ is isomorphic to a subgroup of a symmetric group.
This means that permutations are a sort of universal symbol for
discretizing time series by means of group elements; a different question is
whether such a ``canonical'' embedding is
always the best option in practice.

This being the case, in this paper, we extend two distances in $\mathrm{Sym}%
(L)$, namely, the Cayley distance and the Kendall tau distance (henceforth
called Kendall distance), to arbitrary finite groups via Cayley's theorem. A
possible advantage of the here-proposed distances compared to others (e.g.,
the conventional generator-based distances) is their expediency and
acceptable computation time for groups of moderate cardinality, as happens
in practice. By extension, we discuss also distances for group-valued time
series, which include algebraic representations of time series. This issue
raises naturally when comparing two time series to measure their
``similarity'' (think of classification or
clustering) or studying coupled systems (think of different types of
synchronization). The result is a suite of permutation-based (or ordinal
pattern-based) distances for groups and group-valued time series.

In sum, this is a follow-up paper on the quest to exploit the algebraic
structure of group-valued time series---a possibility rarely used in the
literature. Remarkably, the Cayley and Kendall distances and, hence, their
extensions to general groups, are actually norms of transcripts, which shows
the potential of our algebraic approach. Since our interest in distances
between group elements was motivated by the study of ordinal representations
and transcripts, we will speak of both permutations and ordinal patterns.

To address the aforementioned topics, we begin in Section~\ref{sec:Cayley}
by establishing the mathematical framework, which includes group actions and
group representations. In particular, we will prove Cayley's theorem and
implement it in three different ways---one of them using transcripts. There
and throughout this paper, our approach is formal, the theoretical concepts
being illustrated with simple examples. Section~\ref{sec:OrdPatyEditDist} is
dedicated to the symmetric group and its two standard metrics: the Cayley
and Kendall distances. In Section~\ref{sec:DistforGroups}, we transition from
the symmetric group to general groups and propose a distance based on
Cayley`s Theorem (Section~\ref{sec:DistforGroupsSub1}). This distance is
compared to the conventional string metric for finitely generated groups
(Section~\ref{sec:DistwithGenerators}) in Section~\ref{sec:Discussion}.
Possible extensions to distances between group-valued times series are
discussed in Section~\ref{sec:DistforTS} and illustrated with mathematical
simulations in Section~\ref{sec:NumSimul}. This paper ends with the
conclusions in Section~\ref{sec:Conclusion}.


\section{Groups, Group Actions and Cayley's Theorem}

\label{sec:Cayley}

In this section, we set the mathematical framework of this paper---group actions and group representations~\cite{Herstein1996,Lang2005,Fraleigh2013}.

\begin{definition}
\label{DefGroup} A group $(\mathcal{G},\ast )$ is a nonempty set $\mathcal{G}$
endowed with a binary operation ``$\ast $", sometimes called composition law or product, satisfying the following properties.

\begin{itemize}
\item[\textbf{\emph{(G1)}}] \emph{Associativity:} For all $a,b,c\in \mathcal{G}$, it is true
that $(a\ast b)\ast c=a\ast (b\ast c)$.

\item[\textbf{\emph{(G2)}}] \emph{Identity element:} There exists an element $e\in \mathcal{G%
}$, called the \textit{identity} (or \textit{neutral}) \textit{element},
such that $a\ast e=e\ast a=a$ for all $a\in \mathcal{G}$.

\item[\textbf{\emph{(G3)}}] \emph{Inverse element:} For every $a\in \mathcal{G}$, there
exists an element $a^{-1}\in \mathcal{G}$, called the \textit{inverse element%
} of $a$, such that $a\ast a^{-1}=a^{-1}\ast a=e$.
\end{itemize}
\end{definition}

It can be proved that the identity of a group and the inverse of each
element are unique. Groups whose product is commutative (i.e., $a\ast
b=b\ast a$ for all $a,b\in \mathcal{G}$) are called \textit{commutative} or 
\textit{abelian}. Examples of abelian groups are the real numbers endowed
with addition and the nonzero real numbers endowed with multiplication.
Invertible square matrices are examples of nonabelian groups under
multiplication. If the binary operation is clear from the context, then $(%
\mathcal{G},\ast )$ is shortened to $\mathcal{G}$.

\begin{definition}
\label{DefinitionAction}If $(\mathcal{G},\ast )$ is a group and $S$ a
nonempty set, then a left group action of $\mathcal{G}$ on $S$ is a mapping $%
F:\mathcal{G}\times S\rightarrow S$ such that it satisfies the following two
axioms:

\begin{itemize}
\item[\textbf{\emph{(L1a)}}]~\emph{Identity:} $F(e,s)=s$ for all $s\in S$, where $e$ is the
identity element of $\mathcal{G}$.

\item[\textbf{\emph{(L2a)}}]~\emph{Compatibility:} $F(a,F(b,s))=F(a\ast b,s)$ for all $a,b\in 
\mathcal{G}$ and $s\in S$.
\end{itemize}
\end{definition}

If $F$ is a left action of $\mathcal{G}$ on $S$, we can define the function $%
F_{a}:=F(a,\cdot ):S\rightarrow S$, i.e., 
\begin{equation}
F_{a}(s)=F(a,s)  \label{F_a}
\end{equation}%
for each $a\in \mathcal{G}$. For $F_{a}$, the axioms L1a and L2a read as
follows:

\begin{itemize}
\item[\textbf{(L1b)}]~Identity: \emph{$F_{e}$ is the identity mapping $s\mapsto s$ for all $%
s\in S$.}

\item[\textbf{(L2b)}]~Compatibility: \emph{$F_{a}\circ F_{b}=F_{a\ast b}$ for all $a,b\in 
\mathcal{G}$.}
\end{itemize}

\begin{lemma}
\label{LemmaAutGroup}

\begin{itemize}
\item[(i)] $F_{a}:S\rightarrow S$ is a bijection for each $a\in \mathcal{G}$. 

\item[(ii)] The set $\{F_{a}:S\rightarrow S:a\in \mathcal{G}\}$
endowed with function composition is a group.
\end{itemize}
\end{lemma}

\begin{proof}

\begin{itemize}
\item[(i)] Since $F_{a}$ is defined from $S$ into itself, it suffices to prove that
every $s\in S$ has an inverse. Indeed, $F_{a}^{-1}(s)=F_{a^{-1}}(s)\in S$
because $F_{a}(F_{a^{-1}}(s))=F_{a\ast a^{-1}}(s)=F_{e}(s)=s$ by axioms L2b
and L1b.

\item[(ii)] According to Definition~\ref{DefGroup}, we have to prove three
properties: (G1) associativity is a general property of the composition of
functions; (G2) $F_{e}$ is the identity because of axiom L1b; (G3)\ for all
mappings $F_{a}$, $(F_{a})^{-1}=F_{a^{-1}}$ as in (i).
\end{itemize}
\end{proof}

Bijections from a finite set $S$ onto itself are called \textit{permutations}%
. So, according to Lemma~\ref{LemmaAutGroup}\emph{(i)}
, the mappings $F_{a}$ are
permutations. The permutations on $S$, endowed with function composition,
build a group called the \textit{symmetric group} $\mathrm{Sym}(S)$. In this
paper, we consider only finite groups $\mathcal{G}$ and finite sets $S$, so,
if $\left\vert S\right\vert $ is the cardinality of $S$, then $\left\vert 
\mathrm{Sym}(S)\right\vert =\left\vert S\right\vert !$. Since the properties
of the permutations on $S$ do not depend on $S$ but only on $\left\vert
S\right\vert $, we choose $S=\{1,2,\ldots,\left\vert S\right\vert \}$, unless
otherwise stated, and also refer to $\mathrm{Sym}(S)$ as the symmetric group
of degree $\left\vert S\right\vert $, $\mathrm{Sym}(\left\vert S\right\vert
) $. As a historical note, the symmetric group goes back to \'{E}variste
Galois (1811--1832) and his work on the resolution of algebraic equations by
means of radicals.

Furthermore, Lemma~\ref{LemmaAutGroup}\emph{(ii)} states that the set of
permutations $\{F_{a}:S\rightarrow S:a\in \mathcal{G}\}$ is a subgroup (of
cardinality $\left\vert \mathcal{G}\right\vert $) of $\mathrm{Sym}(S)$. This
result together with axiom L2b, which spells out that the mapping $\Phi
:a\mapsto F_{a}$ preserves the algebraic structure of $\mathcal{G}$, are
merged in the following theorem.

\begin{theorem}
\label{ThmCayley}Any (left) group action $F:\mathcal{G}\times S\rightarrow S$
of a group $\mathcal{G}$ on a finite set $S$ defines a group homomorphism $%
\Phi :a\mapsto F_{a}$ from $\mathcal{G}$ into $\mathrm{Sym}(S)$. Therefore, $%
\Phi $ is a representation of the group $\mathcal{G}$ by means of
permutations $F_{a}:=F(a,\cdot ):S\rightarrow S$.
\end{theorem}

In other words, every group $\mathcal{G}$ is isomorphic to a subgroup $%
\mathcal{H}$ of $\mathrm{Sym}(S)$, namely $\mathcal{H}=\Phi (\mathcal{G})$,
hence, $\left\vert \mathcal{H}\right\vert =\left\vert \mathcal{G}\right\vert 
$. In this formulation, Theorem~\ref{ThmCayley} is known as \textit{Cayley's
theorem}. Therefore, we will call $\Phi :\mathcal{G}\rightarrow $ $\mathrm{%
Sym}(S)$ Cayley's homomorphism and, abusing notation, $\Phi :\mathcal{G}%
\rightarrow $ $\mathcal{H}$ Cayley's isomorphism. Below, we will discuss
three different implementations of Cayley's isomorphism.

To apply Theorem~\ref{ThmCayley}, label the elements of $\mathcal{G}$ with
the conventional set $\{1,2,\ldots,\left\vert \mathcal{G}\right\vert \}$. For
every $a\in \mathcal{G}$, let%
\begin{equation}
F_{a}=\left( 
\begin{array}{ccccc}
1 & \cdots & k & \cdots & \left\vert \mathcal{G}\right\vert \\ 
F_{a}(1) & \cdots & F_{a}(k) & \cdots & F_{a}(\left\vert \mathcal{G}%
\right\vert )%
\end{array}%
\right) =\left( 
\begin{array}{ccccc}
1 & \cdots & k & \cdots & \left\vert \mathcal{G}\right\vert \\ 
n_{1} & \cdots & n_{k} & \cdots & n_{\left\vert \mathcal{G}\right\vert }%
\end{array}%
\right)  \label{2line form}
\end{equation}%
be the matrix (or two-line) form of the permutation $F_{a}$, where $%
(n_{1},\ldots,n_{k},\ldots,n_{\left\vert \mathcal{G}\right\vert })$ is a shuffle
of $(1,2,\ldots,\left\vert \mathcal{G}\right\vert )$. Therefore, every element $%
a\in \mathcal{G}$ can be identified with the\textit{\ one-line form} $%
(n_{1},n_{2},\ldots,n_{\left\vert \mathcal{G}\right\vert })$ of $F_{a}$. In the
numerical examples below, we will juxtapose the components of $%
(n_{1},n_{2},\ldots,n_{\left\vert \mathcal{G}\right\vert })$ and drop the
parentheses for a compact notation.

\begin{remark}
\label{REmarkRightLeft}In addition to left actions of a group $(\mathcal{G}%
,\ast )$ on a finite set $S$, there are also right actions $\tilde{F}%
:S\times \mathcal{G}\rightarrow S$, defined by (R1a) $\tilde{F}(s,e)=s$ for
all $s\in S$, and (R2a) $\tilde{F}(\tilde{F}(s,a),b)=\tilde{F}(s,a\ast b)$,
as well as the corresponding group homomorphism $a\rightarrow \tilde{F}_{a}:=%
\tilde{F}(\cdot ,a)$ from $\mathcal{G}$ to $\mathrm{Sym}(S)$, such that
(R1b) $\tilde{F}_{e}$ is the identity map $s\mapsto s$ for all $s\in S$, and
(R2b) $\tilde{F}_{a}\circ \tilde{F}_{b}=\tilde{F}_{a\ast b}$ for all $a,b\in 
\mathcal{G}$.   The difference between left and right actions is that in the
function composition $F_{a}\circ F_{b}=F_{a\ast b}$ (L2b), $F_{b}$ acts
first on $s\in S$ and $F_{a}$ second (as in the standard convention),
whereas in $\tilde{F}_{a}\circ \tilde{F}_{b}=\tilde{F}_{a\ast b}$ (R2b), $%
\tilde{F}_{a}$ acts first on $s\in S$ and $\tilde{F}_{b}$ second. Henceforth,
we only consider left actions because the binary operation of the symmetric
group, the main character of this paper, is precisely function composition
and so we can use the standard convention.
\end{remark}

There is a particular case of Theorem~\ref{ThmCayley} that is of special
interest here, namely, $S=\mathcal{G}$, i.e., when the group $\mathcal{G}$
acts on itself. In this particular case, we are going to highlight three
implementations of Cayley's isomorphism $\Phi :\mathcal{G}\ni a\mapsto
F_{a}\in \mathrm{Sym}(\mathcal{G})$ via left actions.

\begin{itemize}
\item[\textbf{(A)}] \textit{Left translations}: The mapping $(a,b)\mapsto \Lambda
(a,b)=a\ast b$ is a left action of $\mathcal{G}$ on itself, so 
\begin{equation}
\Lambda _{a}(b)=a\ast b  \label{FreeAction}
\end{equation}%
is a permutation on $\mathcal{G}$ for every $a\in \mathcal{G}$, called a
left translation by $a$.

\item[\textbf{(B)}] \textit{Right translations}: The mapping $(a,b)\mapsto
R(a,b)=b\ast a^{-1}$ is a left action of $\mathcal{G}$ on itself, so%
\begin{equation}
R_{a}(b)=b\ast a^{-1}  \label{Transcript}
\end{equation}%
is a permutation on $\mathcal{G}$ for every $a\in \mathcal{G}$, called a
right translation by $a$. Let us mention that the operation $R(a,b)$ is also
called the \textit{transcription} from the (source) symbol $a$ to the
(target) symbol $b$ in~\cite{Monetti2009}. Note that $\Lambda
_{a}(b)=R_{b^{-1}}(a)$ and $R_{a}(b)=\Lambda _{b}(a^{-1})$.

\item[\textbf{(C)}] \textit{Adjoint actions}: The mapping $(a,b)\mapsto \mathrm{Ad}%
(a,b)=a\ast b\ast a^{-1}$ is a left action of $\mathcal{G}$ on itself, so%
\begin{equation}
\mathrm{Ad}_{a}(b)=a\ast b\ast a^{-1}  \label{Adjoint}
\end{equation}%
is a permutation on $\mathcal{G}$ for every $a\in \mathcal{G}$, called the
adjoint action of $a$.
\end{itemize}

Comparing Equations (\ref{FreeAction})--(\ref{Adjoint}), we conclude that the
implementation (\ref{FreeAction}) of Cayley's isomorphism $\Phi :a\mapsto
F_{a}$ is the most convenient in practice, since the (one-line form of the)
permutations $\Lambda _{a}:b\mapsto a\ast b$ can be read immediately row by
row in the multiplication table of $\mathcal{G}$. Indeed, if $%
\{a_{1},a_{2},\ldots.,a_{\left\vert \mathcal{G}\right\vert }\}$ is an
enumeration of the elements of $\mathcal{G}$, then $\Lambda _{a_{i}}$ is the 
$i$-th row of the multiplication table
 $(a_{i}\ast a_{j})_{1\leq i,j\leq
\left\vert \mathcal{G}\right\vert }$, i.e.,%
\begin{equation}
\Lambda _{a_{i}}=(a_{i}\ast a_{1},\ldots,a_{i}\ast a_{j},\ldots,a_{i}\ast
a_{\left\vert \mathcal{G}\right\vert })=\left( 
\begin{array}{ccccc}
a_{1} & \cdots & a_{j} & \cdots & a_{\left\vert \mathcal{G}\right\vert } \\ 
a_{i}\ast a_{1} & \cdots & a_{i}\ast a_{j} & \cdots & a_{i}\ast
a_{\left\vert \mathcal{G}\right\vert }%
\end{array}%
\right) .  \label{Lambda(a)}
\end{equation}

\begin{example}
Let $\mathcal{G}=\mathrm{Sym}(3)$. By Equation (\ref{FreeAction}), the
isomorphic copies $\Lambda _{\mathbf{r}}\in \mathrm{Sym}(\mathcal{G})=%
\mathrm{Sym}(\mathrm{Sym}(3))=\mathrm{Sym}(6)$ of $\mathbf{r}\in
\{123,132,213,231,312,321\}\mathcal{\ }$are given by the rows of the
``multiplication" table of $\mathrm{Sym}(3) $,
 
\begin{equation}
\begin{tabular}{|c||l|l|l|l|l|l|}
\hline
$\mathbf{r}\circ \mathbf{s}$ & $123$ & $132$ & $213$ & $231$ & $312$ & $321$
\\ \hline\hline
$123$ & $123$ & $132$ & $213$ & $231$ & $312$ & $321$ \\ \hline
$132$ & $132$ & $123$ & $312$ & $321$ & $213$ & $231$ \\ \hline
$213$ & $213$ & $231$ & $123$ & $132$ & $321$ & $312$ \\ \hline
$231$ & $231$ & $213$ & $321$ & $312$ & $123$ & $132$ \\ \hline
$312$ & $312$ & $321$ & $132$ & $123$ & $231$ & $213$ \\ \hline
$321$ & $321$ & $312$ & $231$ & $213$ & $132$ & $123$ \\ \hline
\end{tabular}
\label{mult table Sym(3)}
\end{equation}%
where $\mathbf{r}\circ \mathbf{s}$ stands for the composition of the
permutation $\mathbf{r}$ that labels a row with the permutation $\mathbf{s}$ that labels a column. Therefore, 
\begin{equation}
\begin{tabular}{|c||l|l|l|l|l|l|}
\hline
& $123$ & $132$ & $213$ & $231$ & $312$ & $321$ \\ \hline\hline
$\Lambda _{123}$ & $123$ & $132$ & $213$ & $231$ & $312$ & $321$ \\ \hline
$\Lambda _{132}$ & $132$ & $123$ & $312$ & $321$ & $213$ & $231$ \\ \hline
$\Lambda _{213}$ & $213$ & $231$ & $123$ & $132$ & $321$ & $312$ \\ \hline
$\Lambda _{231}$ & $231$ & $213$ & $321$ & $312$ & $123$ & $132$ \\ \hline
$\Lambda _{312}$ & $312$ & $321$ & $132$ & $123$ & $231$ & $213$ \\ \hline
$\Lambda _{321}$ & $321$ & $312$ & $231$ & $213$ & $132$ & $123$ \\ \hline
\end{tabular}
\label{Sym(Sym(3))F}
\end{equation}%
For example, 
\begin{equation*}
\Lambda _{231}:123\mapsto 231,\;132\mapsto 213,\;213\mapsto 321,\;231\mapsto
312,\;312\mapsto 123,\;321\mapsto 132,
\end{equation*}%
or, in one-line form, $\Lambda _{231}=(231,213,321,312,123,132)$. From 
\begin{equation*}
123^{-1}=123,\;132^{-1}=132,\;213^{-1}=213,\;231^{-1}=312,\;312^{-1}=231,\;321^{-1}=321,
\end{equation*}%
table (\ref{mult table Sym(3)}) and  Equation (\ref{Transcript}), we obtain
similarly that the copies $R_{\mathbf{r}}\in \mathrm{Sym}(\mathrm{Sym}(3))$
of $\mathbf{r}\in \mathrm{Sym}(3)$ via right translations are given by \vspace{12pt}
\begin{equation}
\begin{tabular}{|c||l|l|l|l|l|l|}
\hline
& $123$ & $132$ & $213$ & $231$ & $312$ & $321$ \\ \hline\hline
$R_{123}$ & $123$ & $132$ & $213$ & $231$ & $312$ & $321$ \\ \hline
$R_{132}$ & $132$ & $123$ & $231$ & $213$ & $321$ & $312$ \\ \hline
$R_{213}$ & $213$ & $312$ & $123$ & $321$ & $132$ & $231$ \\ \hline
$R_{231}$ & $312$ & $213$ & $321$ & $123$ & $231$ & $132$ \\ \hline
$R_{312}$ & $231$ & $321$ & $132$ & $312$ & $123$ & $213$ \\ \hline
$R_{321}$ & $321$ & $231$ & $312$ & $132$ & $213$ & $123$ \\ \hline
\end{tabular}
\label{Sym(Sym(3))T}
\end{equation}
\end{example}

\begin{example}
\label{ExampleCyclicGroup}Let $\mathcal{G}=\{\theta ^{0},\theta ^{1},\theta
^{2},\theta ^{3}\}$ endowed with the product $\theta ^{i}\ast \theta
^{j}=\theta ^{j}\ast \theta ^{i}=\theta ^{i+j}$ where, in this example, the
exponents are taken modulo $4$. Hence, $\theta ^{0}$ is the identity and $%
(\theta ^{i})^{-1}=\theta ^{4-i}$. By definition, $\mathcal{G}$ is a cyclic
group generated by the element $\theta ^{1}$. Alternatively, $\mathcal{G}$
can be identified with the additive group $\{0,1,2,3\}$, where the sum is
taken modulo $4$.

(i) The four permutations $\Lambda _{\theta ^{i}}:\theta ^{j}\mapsto \theta
^{i}\ast \theta ^{j}=\theta ^{i+j}$, corresponding to Equation (\ref%
{FreeAction}) under the isomorphism $\Phi :\theta ^{i}\mapsto $ $\Lambda
_{\theta ^{i}}\in \mathrm{Sym}(\mathcal{G})$, are given in the following
table:%
\begin{equation}
\begin{tabular}{|c||c|c|c|c|}
\hline
& $\theta ^{0}$ & $\theta ^{1}$ & $\theta ^{2}$ & $\theta ^{3}$ \\ 
\hline\hline
$\Lambda _{\theta ^{0}}$ & $\theta ^{0}$ & $\theta ^{1}$ & $\theta ^{2}$ & $%
\theta ^{3}$ \\ \hline
$\Lambda _{\theta ^{1}}$ & $\theta ^{1}$ & $\theta ^{2}$ & $\theta ^{3}$ & $%
\theta ^{0}$ \\ \hline
$\Lambda _{\theta ^{2}}$ & $\theta ^{2}$ & $\theta ^{3}$ & $\theta ^{0}$ & $%
\theta ^{1}$ \\ \hline
$\Lambda _{\theta ^{3}}$ & $\theta ^{3}$ & $\theta ^{0}$ & $\theta ^{1}$ & $%
\theta ^{2}$ \\ \hline
\end{tabular}
\label{TableF}
\end{equation}%
So, for instance, the second row of this table spells out%
\begin{equation*}
\Lambda _{\theta ^{1}}:\theta ^{0}\mapsto \theta ^{1},\;\theta ^{1}\mapsto
\theta ^{2},\;\theta ^{2}\mapsto \theta ^{3},\;\theta ^{3}\mapsto \theta
^{0},\;
\end{equation*}%
or $\Lambda _{\theta ^{1}}=(\theta ^{1},\theta ^{2},\theta ^{3},\theta ^{0})$.

(ii) The four permutations $R_{\theta ^{i}}:\theta ^{j}\mapsto \theta
^{j}\ast (\theta ^{i})^{-1}=\theta ^{j-i}$, corresponding to Equation (\ref%
{Transcript}), under the isomorphism $\Phi :\theta ^{i}\mapsto $ $R_{\theta
^{i}}\in \mathrm{Sym}(\mathcal{G})$, are given in the following table:
\begin{equation}
\begin{tabular}{|c||c|c|c|c|}
\hline
& $\theta ^{0}$ & $\theta ^{1}$ & $\theta ^{2}$ & $\theta ^{3}$ \\ 
\hline\hline
$R_{\theta ^{0}}$ & $\theta ^{0}$ & $\theta ^{1}$ & $\theta ^{2}$ & $\theta
^{3}$ \\ \hline
$R_{\theta ^{1}}$ & $\theta ^{3}$ & $\theta ^{0}$ & $\theta ^{1}$ & $\theta
^{2}$ \\ \hline
$R_{\theta ^{2}}$ & $\theta ^{2}$ & $\theta ^{3}$ & $\theta ^{0}$ & $\theta
^{1}$ \\ \hline
$R_{\theta ^{3}}$ & $\theta ^{1}$ & $\theta ^{2}$ & $\theta ^{3}$ & $\theta
^{0}$ \\ \hline
\end{tabular}
\label{TableT}
\end{equation}%
So, if in table (\ref{TableF}), $\Lambda _{\theta ^{i+1}}$ is obtained from $%
\Lambda _{\theta ^{i}}$ by a clockwise (negative) circular shift, in table (\ref{TableT}), the circular shift to obtain $R_{\theta ^{i+1}}$ from $%
R_{\theta ^{i}}$ is counterclockwise (positive).
\end{example}


\section{Ordinal Patterns and Distances}

\label{sec:OrdPatyEditDist}

In the previous, section we have focused on group actions and the embedding
of a group in a symmetric group. What is still missing is metric tools that
can further boost applications in the realm of group-valued time series.
Since the motivation and objective of this paper are the applications of
such tools to symbolic representations of time series via group elements, we
begin this section by briefly explaining how such symbolic representations
arise in time series analysis. The choice of ordinal patterns (or
permutations) responds to the popularity of these symbols among time series
analysts. Then, we introduce the concept of distance in the symmetric group
and, in the next section, we do the same for general groups.

\subsection{Ordinal Patterns}

\label{sec:OrdPat}

Symmetric groups are very popular for symbolic representations since the
concept of \textit{ordinal pattern} was introduced in~\cite{Bandt2002}.
Given a real-valued time series $x=(x_{t})_{t\geq 0}$, an \textit{ordinal
representation} of $x$ is a symbolic time series $(\mathbf{r}_{t})_{t\geq 0}$
whose alphabet is $\mathrm{Sym}(L)$, the symmetric group of degree $L\geq 2$%
. How are the permutations $\mathbf{r}_{t}$ obtained from $x$? Let $%
x_{t}^{L}:=x_{t},x_{t+1},\ldots,x_{t+L-1}$ be a window (segment, sequence,
block, \ldots) of size $L$. Then, $\mathbf{r}_{t}=(r_{1},r_{2},\ldots,r_{L})$ is
the rank vector of $x_{t}^{L}$, that is, $(r_{1},r_{2},\ldots,r_{L})$ is the
permutation of $\{1,2,\ldots,L\}$ such that%
\begin{equation}
x_{t+r_{1}-1}<x_{t+r_{2}-1}<\ldots<x_{t+r_{L}-1}.  \label{ord pat}
\end{equation}%
In other words, the rank vector $\mathbf{r}_{t}$ is viewed as the one-line
form of the permutation $1\mapsto r_{1}$, $2\mapsto r_{2}$, \ldots, $L\mapsto
r_{L}$, i.e., $\mathbf{r}_{t}(k)=r_{k}$ for $1\leq k\leq L$. As a matter of
fact, any total ranking can be viewed as a permutation. In case of a tie $%
x_{i}=x_{j}$, one can apply the convention that $x_{i}<x_{j}$ if $i<j$.
Another possibility, more recommended in case of many ties, is to add a
small-amplitude noise to $x_{i}$ and $x_{j}$ to undo the tie. \ As way of
illustration, if $L=4$ and $x_{t}^{L}=2.1,\,0.3,\,1.5,\,2.4$, then $\mathbf{r%
}_{t}=(2,3,1,4)$, or $\mathbf{r}_{t}=2314$ for short.

In~\cite{Bandt2002}, the permutations $\mathbf{r}_{t}$ were called order (or
ordinal) patterns of length $L$, which is the usual name of the
symbols $\mathbf{r}_{t}$ in time series analysis. In addition to the length $%
L$ of the patterns, ordinal representations depend also on a second
parameter: a possible time delay in Equation (\ref{ord pat}). In this paper,
the time delay is set equal to 1 throughout.

{As a side note, the concept of ordinal pattern has been generalized in several directions. Thus, it has been extended to multivariate time series in~\cite{Mohr2020,Amigo2013B}. Spatial ordinal patterns were introduced in~\cite{Ribeiro2012} to analyze two-dimensional images and applied in~\cite{Zunino2016,Bandt2023} to distinguish textures.}


\subsection{Distances for Ordinal Patterns}

\label{sec:EditDist}

In this section, we introduce the Cayley and Kendall distances for the
symmetric group $\mathrm{Sym}(L)$; see~\cite{Deza1998} for a survey about
distances on permutations. We remind first about the concept of distance.

\begin{definition}
\label{DefDistance}Given a nonempty set $S$, a distance is a function $%
d:S\times S\rightarrow \mathbb{R}$ that satisfies the following three axioms
for all points $x,y,z\in S$.

\begin{itemize}
\item[\textbf{\emph{(D1)}}] \emph{Positivity:} $d(x,y>0$ and $d(x,y)=0$ if and only if $x=y.$

\item[\textbf{\emph{(D2)}}] \emph{Symmetry:} $d(x,y)=d(y,x)$.

\item[\textbf{\emph{(D3)}}] \emph{Triangular inequality: }$d(x,z)\leq d(x,y)+d(y,z)$.
\end{itemize}
\end{definition}

Following the notation in Section~\ref{sec:OrdPat} for ordinal patterns, the
permutations of $\mathrm{Sym}(L)$ will be written in the one-line form $%
\mathbf{r}=(r_{1},r_{2},\ldots,r_{L})$ (possibly shortened to $%
r_{1},r_{2},\ldots,r_{L}$ in numerical examples), where $\mathbf{r}(i)=r_{i}$. If,
furthermore, $\mathbf{s}=(s_{1},s_{2},\ldots,s_{L})\in \mathrm{Sym}(L)$, then $%
\mathbf{r}\circ \mathbf{s}$ is the usual function composition $(\mathbf{r}%
\circ \mathbf{s})(i)=\mathbf{r}(\mathbf{s}(i))$, i.e.,%
\begin{equation}
\mathbf{r}\circ \mathbf{s}=(r_{1},\ldots,r_{k},\ldots,r_{L})\circ
(s_{1},\ldots,s_{k},\ldots,s_{L})=(r_{s_{1}},\ldots,r_{s_{k}},\ldots,r_{s_{L}}),
\label{r rosquilla s}
\end{equation}%
as exemplified in Equation (\ref{mult table Sym(3)}) for $L=3$. Due to the
positivity and symmetry properties of a distance, the $L!\times L!$ \textit{%
distance matrix} $(d(\mathbf{r},\mathbf{s}):\mathbf{r},\mathbf{s}\in \mathrm{%
Sym}(L))$ is symmetric, with $0$'s along the diagonal.

If $\{i_{1},i_{2},\ldots,i_{m}\}\subset \{1,2,\ldots,L\}$, then $%
(i_{1},i_{2},\ldots, i_{m})$ denotes the permutation%
\begin{equation}
i_{1}\mapsto i_{2},\;i_{2}\mapsto i_{3},\;\ldots,\;i_{m-1}\mapsto
i_{m},\;i_{m}\mapsto i_{1},\;  \label{cycle}
\end{equation}%
called a \textit{cycle of length} $m$, $1\leq m\leq L$, or simply an $m$%
-cycle. The notation calls for a warning at this point: do not confuse the
permutation $i_{1},i_{2}\ldots, i_{m}=(i_{1},i_{2},\ldots,i_{m})$ with the cycle $%
(i_{1},i_{2},\ldots,i_{m})$. Every permutation can be written as a product of
disjoint cycles, which is unique except for the order of the factors. For
example, the cycle factorization of the permutation $426135$ is $%
(14)(2)(356) $ or $(14)(356)$ if $1$-cycles (``fixed
elements'') are omitted.

Cycles of length 2 are called \textit{transpositions}. That is, a
transposition is a permutation $\mathbf{t}_{ij}\in \mathrm{Sym}(L)$ such
that $\mathbf{t}_{ij}(i)=j$, $\mathbf{t}_{ij}(j)=i$, and $\mathbf{t}%
_{ij}(k)=k$ for all $k\neq i,j$. If $\mathbf{r}=(r_{1},\ldots,r_{L})$, then%
\begin{equation}
\mathbf{r\circ t}%
_{ij}=(r_{1},\ldots,r_{i-1},r_{j},r_{i+1},\ldots,r_{j-1},r_{i},r_{j+1},\ldots,r_{L}).
\label{transposition}
\end{equation}%
If $\left\vert i-j\right\vert =1$, then $\mathbf{t}_{ij}$ is called an 
\textit{adjacent transposition}. Unlike the factorization of permutations into disjoint cycles, the factorization of permutations into adjacent transpositions (and, hence, into transpositions) is not unique, although the minimal number of factors is. For example, $321 = (12)(23)(12) = (23)(12)(23)$.

\begin{definition}[\cite{Nguyen2024,Kendall1938}]\label{DefDistances}
 Let $\mathbf{r},\mathbf{s}%
\in \mathrm{Sym}(L)$. (a) The \emph{Cayley distance} between the two
permutations $\mathbf{r}$ and $\mathbf{s}$, denoted by $d_{C}(\mathbf{r},%
\mathbf{s})$, is defined as the minimum number of transpositions needed to
transform $\mathbf{r}$ into $\mathbf{s}$. (b) The \emph{Kendall distance}
(also known as the bubble-sort distance) between $\mathbf{r}$ and $\mathbf{s}
$, denoted by $d_{K}(\mathbf{r},\mathbf{s})$, is defined as the minimum
number of \emph{adjacent} transpositions needed to transform $\mathbf{r}$
into $\mathbf{s}$.
\end{definition}

The Cayley and Kendall distances are examples of edit distances between two strings of symbols, which measure the minimum cost sequence of allowed edit operations to transform one string into the other. The use of edit distances to measure the distance between permutations was proposed in~\cite{Sorensen2007}. By definition,
\begin{equation}
d_{C}(\mathbf{r},\mathbf{s})\leq d_{K}(\mathbf{r},\mathbf{s})
\label{d_C <= d_K}
\end{equation}%
for all $\mathbf{r},\mathbf{s}\in \mathrm{Sym}(L)$.

The proofs of the positivity and symmetry (properties (D1) and (D2) in Definition~\ref{DefDistance}) for $d_{C}(\mathbf{r},\mathbf{s})$ and $d_{K}(%
\mathbf{r},\mathbf{s})$ are straightforward. The triangular inequality can
be easily proved by graph-based methods since the permutations of $\mathrm{%
Sym}(L)$ build a connected undirected graph where the nodes (or vertices) correspond to
permutations and the links (or edges) to transpositions. For example, in the case of $%
d_{K}(\mathbf{r},\mathbf{s})$: (i) every node $\mathbf{r}$ is connected to
exactly $L-1$ nearest neighbors, namely, those permutations that differ from 
$\mathbf{r}$ due to transpositions of the adjacent symbols $r_{i},r_{i+1}$
for $1\leq i\leq L-1$, and, hence, (ii) for any two nearest nodes $\mathbf{u}
$ and $\mathbf{v}$, $d_{K}(\mathbf{u},\mathbf{v})=d_{K}(\mathbf{v},\mathbf{u}%
)=1$. Therefore, $d_{K}(\mathbf{r},\mathbf{s})$ counts the number of links
of the shortest path connecting the nodes $\mathbf{r}$ and $\mathbf{s}$. In
other words, each node has degree $L-1$ and all its nearest neighbors (one
link apart) are at distance $1$. The diameter of the graph, i.e., the
farthest distance between any two nodes, corresponds to $\mathbf{r}%
=(r_{1},r_{2},\ldots,r_{L})$ and the order reversing permutation $\mathbf{s}%
=(r_{L},r_{L-1},\ldots,r_{1})$, hence%
\begin{equation}
d_{K,\max }(L)=(L-1)+(L-2)+\ldots+1=\frac{L(L-1)}{2}.  \label{Lmax}
\end{equation}%
Such graphs are called \textit{adjacency graphs} or networks.

Figures~\ref{Figure1} and~\ref{Figure5} show the adjacency graphs of the groups $%
\mathrm{Sym}(3)$ (a cycle in this case) and $\mathrm{Sym}(4)$, respectively.
Unlike the adjacency graphs for the Kendall distance, the adjacency graphs
for the Cayley distance are in general nonplanar, i.e., they have
edge crossings (even for $\mathrm{Sym}(3)$), so we will not use them.

\begin{figure}[H]
\centering
\includegraphics[width=40mm]{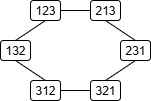}
\caption{Kendall adjacency graph of $\text{Sym}(3)$. A link between two
nodes means that the corresponding permutations differ by an adjacent
transposition, i.e., the Kendall distance between them is 1.}
\label{Figure1}
\end{figure}

\begin{figure}[H]
\centering
\includegraphics[width=135mm]{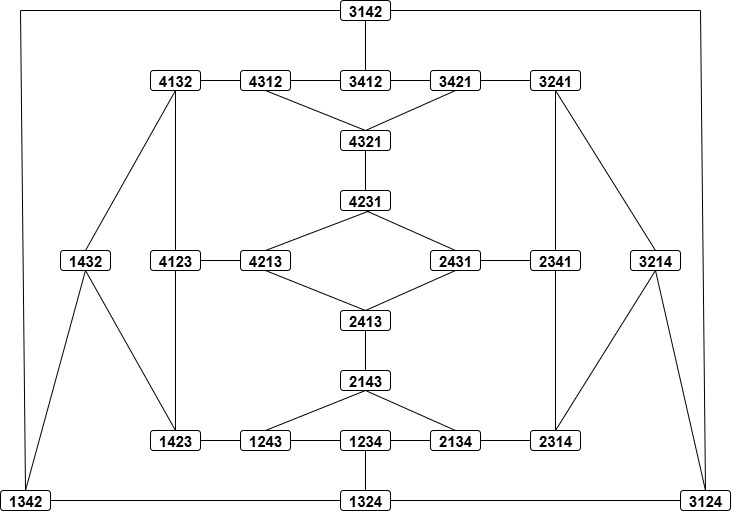}
\caption{Kendall adjacency graph of $\text{Sym}(4)$. A link between two
permutations means that the Kendall distance between them is 1.}
\label{Figure5}
\end{figure}

In the following, whenever convenient for economy of notation, we denote by $%
d_{C,K}$ both the Cayley and Kendall distances.

\begin{proposition}[Invariance of $d_{C,K}$ under left translations]\label{PropInvariance} 
Given $\mathbf{r},\mathbf{s}\in \mathrm{Sym}(L)$, then%
\begin{equation}
d_{C,K}(\mathbf{r},\mathbf{s})=d_{C,K}(\mathbf{u}\circ \mathbf{r},\mathbf{u}%
\circ \mathbf{s})  \label{left invariance}
\end{equation}%
for all $\mathbf{u}\in \mathrm{Sym}(L)$.
\end{proposition}

\begin{proof}
Suppose $d_{C,K}(\mathbf{r},\mathbf{s})=k$, i.e., $k$ is the mimimum number
of transpositions or adjacent transpositions $\mathbf{t}_{i_{1}j_{1}},%
\mathbf{t}_{i_{2}j_{2}},\ldots,\mathbf{t}_{i_{k}j_{k}}\in \mathrm{Sym}(L)$ such
that%
\begin{equation*}
\mathbf{r}=(\ldots((\mathbf{s}\circ \mathbf{t}_{i_{1}j_{1}})\circ \mathbf{t}%
_{i_{2}j_{2}})\circ \ldots \circ \mathbf{t}_{i_{k-1}j_{k-1}})\circ \mathbf{t}%
_{i_{k}j_{k}},
\end{equation*}%
see Equation (\ref{transposition}). Then,%
\begin{equation*}
\mathbf{u}\circ \mathbf{r}=(\ldots((\mathbf{u}\circ \mathbf{s}\circ \mathbf{t}%
_{i_{1}j_{1}})\circ \mathbf{t}_{i_{2}j_{2}})\circ \ldots \circ \mathbf{t}%
_{i_{k-1}j_{k-1}})\circ \mathbf{t}_{i_{k}j_{k}},
\end{equation*}%
which proves that $d_{C,K}(\mathbf{u}\circ \mathbf{r},\mathbf{u}\circ 
\mathbf{s})=k.$
\end{proof}

Since $d_{C,K}(\mathbf{r},\mathbf{s})=d_{C,K}(\mathbf{s},\mathbf{r})$, then $%
d_{C,K}(\mathbf{u}\circ \mathbf{r},\mathbf{u}\circ \mathbf{s})=d_{C,K}(%
\mathbf{u}\circ \mathbf{s},\mathbf{u}\circ \mathbf{r}).$ Choose $\mathbf{u}=%
\mathbf{r}^{-1}$ or $\mathbf{u}=\mathbf{s}^{-1}$ in Equation (\ref{left
invariance}) to prove:

\begin{corollary} \label{PropInvariance2}
For every $\mathbf{r},\mathbf{s}\in \mathrm{Sym}(L)$, 
\begin{equation}
d_{C,K}(\mathbf{r},\mathbf{s})=d_{C,K}(\mathbf{e},\mathbf{r}^{-1}\circ 
\mathbf{s})=d_{C,K}(\mathbf{e},\mathbf{s}^{-1}\circ \mathbf{r}),
\label{d invariance}
\end{equation}%
where $\mathbf{e}$ is the identity permutation.
\end{corollary}

\begin{remark}
\label{RemarkAdmissibleValues1}Owing to Equation (\ref{d invariance}), all
possible values of $d_{C,K}(\mathbf{r},\mathbf{s})$ appear on the row $%
(d_{C,K}(\mathbf{e},\mathbf{u}):\mathbf{u}\in \mathrm{Sym}(L))$ of the
distance matrix.
\end{remark}

Equation (\ref{d invariance}) allows to define in $\mathrm{Sym}(L)$ an
analogue to the concept of norm in a vector space.

\begin{definition}
\label{RemarkNorm}The norm $\left\Vert \cdot \right\Vert _{C,K}$ of $\mathbf{%
r}\in \mathrm{Sym}(L)$ is defined as%
\begin{equation}
\left\Vert \mathbf{r}\right\Vert _{C,K}=d_{C,K}(\mathbf{e},\mathbf{r}).
\label{norm}
\end{equation}
\end{definition}

Then, by Equation (\ref{d invariance}), 
\begin{equation}
d_{C,K}(\mathbf{r},\mathbf{s})=\left\Vert \mathbf{r}^{-1}\circ \mathbf{s}%
\right\Vert _{C,K}=\left\Vert \mathbf{s}^{-1}\circ \mathbf{r}\right\Vert
_{C,K}.  \label{norm2}
\end{equation}

\begin{remark}
The right translation of $b\in \mathcal{G}$ by $a\in \mathcal{G}$, or the
transcript from (the source) $a$ to (the target) $b$, was defined in
Equation (\ref{Transcript}) as $R(a,b)=b\ast a^{-1}$. In view of Equation (%
\ref{norm2}), we conclude that the distance $d_{C,K}(\mathbf{r},\mathbf{s})$
is the norm $\left\Vert \cdot \right\Vert _{C,K}$ of the right translations
or transcripts $R(\mathbf{s}^{-1},\mathbf{r}^{-1})=\mathbf{r}^{-1}\circ 
\mathbf{s}$ and $\ R(\mathbf{r}^{-1},\mathbf{s}^{-1})=\mathbf{s}^{-1}\circ 
\mathbf{r}=$\ $R(\mathbf{s}^{-1},\mathbf{r}^{-1})^{-1}$.
\end{remark}

Corollary~\ref{PropInvariance2} is instrumental for the computation of the
Cayley and Kendall distances~\cite{Nguyen2024}.

\begin{proposition}
\label{PropComput dC dK}(a) Let $\mathbf{u}=(u_{1},\ldots,u_{L})\in \mathrm{Sym}%
(L)$ and $C(\mathbf{u})$ the number of cycles (including $1$-cycles) in the
cycle factorization of the permutation $\mathbf{u}$. Then,%
\begin{equation}
d_{C}(\mathbf{r},\mathbf{s})=L-C(\mathbf{r}^{-1}\circ \mathbf{s})=L-C(%
\mathbf{s}^{-1}\circ \mathbf{r})  \label{Comput d_C}
\end{equation}%
for all $\mathbf{r},\mathbf{s}\in \mathrm{Sym}(L)$

(b) Let $I(\mathbf{u})$ be the number of inversions in the permutation $%
\mathbf{u}$, i.e., the number of ordered pairs $(u_{i},u_{j})$, $1\leq
i<j\leq L$, such that $u_{i}>u_{j}$. Then,%
\begin{equation}
d_{K}(\mathbf{r},\mathbf{s})=I(\mathbf{r}^{-1}\circ \mathbf{s})=I(\mathbf{s}%
^{-1}\circ \mathbf{r})  \label{Comput d_K}
\end{equation}%
for all $\mathbf{r},\mathbf{s}\in \mathrm{Sym}(L)$.
\end{proposition}

From Equation (\ref{Comput d_C}), it follows
\begin{equation}
d_{C}(\mathbf{r},\mathbf{s})\in \{0,1,\ldots,d_{C,\max }(L)\},\text{\ where }%
d_{C,\max }(L)=L-1,  \label{Range d_C}
\end{equation}%
and, according to Equation (\ref{Lmax}), 
\begin{equation}
d_{K}(\mathbf{r},\mathbf{s})\in \left\{ 0,1,\ldots,d_{K,\max }(L)\right\} ,%
\text{\ where }d_{K,\max }(L)=\frac{L(L-1)}{2}.  \label{Range d_K}
\end{equation}

\begin{example}
We illustrate Proposition~\ref{PropComput dC dK} with $L=6$, $\mathbf{r}%
=462531$ and $\mathbf{s} = 236514$. Then,%
\begin{equation*}
\mathbf{s}^{-1}=\left( 
\begin{array}{cccccc}
1 & 2 & 3 & 4 & 5 & 6 \\ 
2 & 3 & 6 & 5 & 1 & 4%
\end{array}%
\right) ^{-1}=\left( 
\begin{array}{cccccc}
1 & 2 & 3 & 4 & 5 & 6 \\ 
5 & 1 & 2 & 6 & 4 & 3%
\end{array}%
\right) ,
\end{equation*}%
so that%
\begin{equation*}
\mathbf{s}^{-1}\circ \mathbf{r}=512643\circ 462531=631425,
\end{equation*}%
whose cycle factorization is 
\begin{equation*}
\mathbf{s}^{-1}\circ \mathbf{r}=(16523)(4).
\end{equation*}%
According to Equation (\ref{Comput d_C}), 
\begin{equation}
d_{C}(\mathbf{r},\mathbf{s})=L-C(\mathbf{s}^{-1}\circ \mathbf{r})=6-2=4.
\label{Ex_dC}
\end{equation}%
As for Equation (\ref{Comput d_K}), the inversions of $\mathbf{s}^{-1}\circ 
\mathbf{r}$ are 
\begin{eqnarray*}
&&(6,3),(6,1),(6,4),(6,2),(6,5), \\
&&(3,1),(3,2), \\
&&(4,2),
\end{eqnarray*}%
so that, 
\begin{equation}
d_{K}(\mathbf{r},\mathbf{s})=I(\mathbf{s}^{-1}\circ \mathbf{r})=8.
\label{EX_dK}
\end{equation}%
Let us check the results (\ref{Ex_dC}) and (\ref{EX_dK}). First, the
transpositions needed to transform $\mathbf{r}$ into $\mathbf{s}$ are the
following:%
\begin{equation*}
\begin{tabular}{lllll}
$\mathbf{r}=\mathbf{4}6\mathbf{2}531$ & $\overset{(13)}{\longrightarrow }$ & 
$2\mathbf{6}45\mathbf{3}1$ & $\overset{(25)}{\longrightarrow }$ & $23\mathbf{%
4}5\mathbf{6}1$ \\ 
& $\overset{(35)}{\longrightarrow }$ & $2365\mathbf{41}$ & $\overset{(56)}{%
\longrightarrow }$ & $236514=\mathbf{s}$%
\end{tabular}%
\end{equation*}%
where the elements being swapped in each transposition have been boldfaced.
Therefore, $d_{C}(\mathbf{r},\mathbf{s})=4$. To check Equation (\ref{EX_dK}%
), call $\delta _{1}$ the number of adjacent transpositions needed to move
in $\mathbf{r}$ the symbol $2$ (the first or leftmost symbol of the target $%
\mathbf{s}$) to the first position; call $\mathbf{r}^{(1)}$ the result.
Similarly, call $\delta _{2}$ the number of adjacent transpositions needed
to move in $\mathbf{r}^{(1)}$ the symbol $3$ (the second symbol of the
target $\mathbf{s}$) to the second position. Proceed analogously until $%
\mathbf{r}^{(k)}=\mathbf{s}$. The adjacent transpositions needed to
transform $\mathbf{r}$ into $\mathbf{s}$ in this example are the following:
\begin{equation*}
\begin{tabular}{lllllll}
$\mathbf{r}=46\mathbf{2}531$ & $\overset{\delta _{1}=2}{\longrightarrow }$ & 
$\mathbf{r}^{(1)}=2465\mathbf{3}1$ & $\overset{\delta _{2}=3}{%
\longrightarrow }$ & $\mathbf{r}^{(2)}=234\mathbf{6}51$ &  &  \\ 
& $\overset{\delta _{3}=1}{\longrightarrow }$ & $\mathbf{r}^{(3)}=2364%
\mathbf{5}1$ & $\overset{\delta _{4}=1}{\longrightarrow }$ & $\mathbf{r}%
^{(4)}=23654\mathbf{1}$ & $\overset{\delta _{5}=1}{\longrightarrow }$ & $%
\mathbf{r}^{(5)}=236514=\mathbf{s}$%
\end{tabular}
\end{equation*}%
where the element of $\mathbf{r}^{(i)}$ ($\mathbf{r}^{(0)}:=\mathbf{r}$)
being moved to the $(i+1)$-position has been boldfaced. This shows that $d_{K}(\mathbf{r},\mathbf{s})=\delta _{1}+\ldots+\delta _{5}=8$.
\end{example}

\begin{example}
\label{ExampleSym(3)}{According to~\cite{Bandt2019}, $\mathcal{G}=\mathrm{Sym}(3)$ is the most
common ordinal representation in data analysis.} The Cayley and Kendall distance matrices for the group $%
\mathrm{Sym}(3)$, Equation (\ref{mult table Sym(3)}), are shown in the tables
 
\begin{equation}
\begin{tabular}{|c||c|c|c|c|c|c|}
\hline
$d_{C}(\mathbf{r},\mathbf{s})$ & $123$ & $132$ & $213$ & $231$ & $312$ & $%
321 $ \\ \hline\hline
$123$ & $0$ & $1$ & $1$ & $2$ & $2$ & $1$ \\ \hline
$132$ & $1$ & $0$ & $2$ & $1$ & $1$ & $2$ \\ \hline
$213$ & $1$ & $2$ & $0$ & $1$ & $1$ & $2$ \\ \hline
$231$ & $2$ & $1$ & $1$ & $0$ & $2$ & $1$ \\ \hline
$312$ & $2$ & $1$ & $1$ & $2$ & $0$ & $1$ \\ \hline
$321$ & $1$ & $2$ & $2$ & $1$ & $1$ & $0$ \\ \hline
\end{tabular}
\label{d_C(Sym(3))}
\end{equation}%
and 
\begin{equation}
\begin{tabular}{|c||c|c|c|c|c|c|}
\hline
$d_{K}(\mathbf{r},\mathbf{s})$ & $123$ & $132$ & $213$ & $231$ & $312$ & $%
321 $ \\ \hline\hline
$123$ & $0$ & $1$ & $1$ & $2$ & $2$ & $3$ \\ \hline
$132$ & $1$ & $0$ & $2$ & $3$ & $1$ & $2$ \\ \hline
$213$ & $1$ & $2$ & $0$ & $1$ & $3$ & $2$ \\ \hline
$231$ & $2$ & $3$ & $1$ & $0$ & $2$ & $1$ \\ \hline
$312$ & $2$ & $1$ & $3$ & $2$ & $0$ & $1$ \\ \hline
$321$ & $3$ & $2$ & $2$ & $1$ & $1$ & $0$ \\ \hline
\end{tabular}
\label{d_K(Sym(3))}
\end{equation}%
As shown in Equations (\ref{d_C <= d_K}), (\ref{Range d_C}) and (\ref
{Range d_K}), $d_{C}(\mathbf{r},\mathbf{s})\leq d_{K}(\mathbf{r},\mathbf{s})$
for all $\mathbf{r},\mathbf{s\in }\mathrm{Sym}(3)$, $d_{C}(\mathbf{r},%
\mathbf{s})\in \{0,1,2\}$ and $d_{K}(\mathbf{r},\mathbf{s})\in \{0,1,2,3\}$.

{Owing to their large size, the Cayley and Kendall distance matrices for $%
\mathcal{G}=\mathrm{Sym}(4)$ have been moved to Appendix \ref{appendixa}. In this case, $%
d_{C}(\mathbf{r},\mathbf{s})\in \{0,1,2,3\}$, and $d_{K}(\mathbf{r},\mathbf{s%
})\in \{0,1,2,3,4,5,6\}$. Needless to say, the distances $d_{K}(\mathbf{r},%
\mathbf{s})$ in tables (\ref{d_K(Sym(3))}) and \ref{tab:table02} can be easily checked
in the corresponding adjacency graphs, Figures~\ref{Figure1} and \ref%
{Figure5}, where each link stands for distance 1.}

\end{example}


\section{Distances for General Groups}

\label{sec:DistforGroups}

In the first part of this section, we harness Cayley's theorem to transport
the Cayley and Kendall distances in $\mathrm{Sym}(L)$ (or, for that matter,
any distance defined in $\mathrm{Sym}(L)$) to any finite group $(\mathcal{G}%
,\ast )$ with $\left\vert \mathcal{G}\right\vert =L$. In the second part, we
briefly introduce the distance with respect to a generating system. We also
discuss the advantages of the first approach as compared to the second.


\subsection{Permutation-Based Distance for Groups}

\label{sec:DistforGroupsSub1}

Let $\Phi :\mathcal{G}\rightarrow \mathcal{H}$ be Cayley's isomorphism,
where $\mathcal{H}$ is a subgroup of $\mathrm{Sym}(\mathcal{G})$ (namely, $%
\mathcal{H}=$ $\Phi (\mathcal{G})$) with $\left\vert \mathcal{H}\right\vert
=\left\vert \mathcal{G}\right\vert $). This means:

\begin{itemize}
\item[\textbf{(i)}] $\Phi (e)$ $=(1,2,\ldots,\left\vert \mathcal{G}\right\vert )$, where $e$ is the identity of $\mathcal{G}$.

\item[\textbf{(ii)}] $\Phi (a\ast b)$ $=\Phi (a)\circ \Phi (b)$ for all $a,b\in \mathcal{G}$. Hence, $\Phi (a^{-1})=\Phi (a)^{-1}$.
\end{itemize}

To endow $\mathcal{G}$ with a distance, we transport the distance $d_{C,K}(\mathbf{r},\mathbf{s})$ from the group $\Phi (\mathcal{G})\subset \mathrm{Sym}(\mathcal{G})$ to $\mathcal{G}$ and promote $\Phi $ to an isometry.

\begin{definition}
\label{DefDist(a,b)}Let $\Phi $ be the Cayley isomorphism for a finite group 
$\mathcal{G}$. Then, $D_{C,K}^{(\Phi )}$ is the distance in $\mathcal{G}$
defined as 
\begin{equation}
D_{C,K}^{(\Phi )}(a,b)=d_{C,K}(\Phi (a),\Phi (b)).  \label{dist(a,b)}
\end{equation}
\end{definition}

Therefore, $D_{C,K}^{(\Phi )}$ has the same properties as $d_{C,K}$. In
particular:

\begin{itemize}
\item \textit{Left invariance}: By Equation (\ref{left invariance}),%
\begin{equation}
D_{C,K}^{(\Phi )}(a,b)=D_{C,K}^{(\Phi )}(c\ast a,c\ast b)  \label{Left inv D}
\end{equation}%
for all $a,b,c\in \mathcal{G}$, hence,%
\begin{equation}
D_{C,K}^{(\Phi )}(a,b)=D_{C,K}^{(\Phi )}(e,a^{-1}\ast b)=D_{C,K}^{(\Phi
)}(e,b^{-1}\ast a),  \label{Left inv D 2}
\end{equation}%
where $e$ is the identity of $\mathcal{G}$.

\item \textit{Norm-based definition}: By Equation (\ref{norm2}),%
\begin{equation}
D_{C,K}^{(\Phi )}(a,b)=\left\Vert \Phi (a)^{-1}\circ \Phi (b)\right\Vert
_{C,K}=\left\Vert \Phi (b)^{-1}\circ \Phi (a)\right\Vert _{C,K},
\label{dist(a,b) 2}
\end{equation}%
where $\left\Vert \cdot \right\Vert _{C,K}$ is the Cayley/Kendall norm in $%
\mathrm{Sym}(\mathcal{G})$, i.e.,%
\begin{equation}
\left\Vert \mathbf{r}\right\Vert _{C,K}=d_{C,K}(\mathbf{e},\mathbf{r})
\label{norm C,K}
\end{equation}%
for all $\mathbf{r}\in \mathrm{Sym}(\mathcal{G})$, $\mathbf{e}$ being the
identity of $\mathrm{Sym}(\mathcal{G}).$
\end{itemize}

From Equations (\ref{d_C <= d_K}) and (\ref{dist(a,b)}), it follows%
\begin{equation}
D_{C}^{(\Phi )}(a,b)\leq D_{K}^{(\Phi )}(a,b)  \label{D_C<=D_K}
\end{equation}%
for all $a,b\in \mathcal{G}$, since $\Phi (a),\Phi (b)\in \mathrm{Sym}(%
\mathcal{G})$. Furthermore, by Equation (\ref{Range d_C}),%
\begin{equation}
D_{C}^{(\Phi )}(a,b)\in \{0,1,\ldots,D_{C,\max }^{(\Phi )}(\left\vert \mathcal{G%
}\right\vert )\},\;\text{where }D_{C,\max }^{(\Phi )}(\left\vert \mathcal{G}%
\right\vert )=\left\vert \mathcal{G}\right\vert -1,  \label{DC max}
\end{equation}%
and, by Equation (\ref{Range d_K}), 
\begin{equation}
D_{K}^{(\Phi )}(a,b)\in \left\{ 0,1,\ldots,D_{K,\max }^{(\Phi )}(\left\vert 
\mathcal{G}\right\vert )\right\} ,\;\text{where }D_{K,\max }^{(\Phi
)}(\left\vert \mathcal{G}\right\vert )=\frac{\left\vert \mathcal{G}%
\right\vert (\left\vert \mathcal{G}\right\vert -1)}{2}.  \label{DK max}
\end{equation}

\begin{remark}
\label{RemarkGaps}In the case $\mathcal{G}=\mathrm{Sym}(L)$ of Section \ref%
{sec:EditDist}, the distances $d_{C,K}(\mathbf{r},\mathbf{s})$ take on all
integer values ranging from $0$ to their respective maxima $d_{C,\max }=L-1$
(Equation (\ref{Range d_C})), and $d_{K,\max }=L(L-1)/2$ (Equation (\ref%
{Range d_K})); think of the corresponding adjacency graphs. However, this
does not happen with $D_{C,K}^{(\Phi )}(a,b)$ because $\Phi (\mathcal{G})$
is a subgroup of cardinality $\left\vert \mathcal{G}\right\vert $ of the
group $\mathrm{Sym}(\mathcal{G})$, whose cardinality is $\left\vert \mathcal{%
G}\right\vert !$, so not all possible distances can be realized (unless $%
\left\vert \mathcal{G}\right\vert =2$). We call ``forbidden
distances for $D_{C,K}^{(\Phi )}$" the values in $%
\{0,1,\ldots,D_{C,K,\max }^{(\Phi )}\}$ that are missing in the adjacency
subgraph of $\Phi (\mathcal{G})$; otherwise, they are called allowed or
admissible distances. By Equation (\ref{Left inv D 2}) (or \mbox{Remark \ref%
{RemarkAdmissibleValues1}}), the admissible distances for $D_{C,K}^{(\Phi )}$
can be read in the row $(D_{C,K}^{(\Phi )}(e,c):c\in \mathcal{G}))$ of the
distance matrix.
\end{remark}

In general, the definition (\ref{dist(a,b)}) depends on the implementation
of Cayley's isomorphism $\Phi $, e.g., whether $\Phi (a)$ is (i) a left
translation $\Lambda _{a}$ (Equation~(\ref{FreeAction})), (ii) a right
translation $R_{a}$ (Equation (\ref{Transcript})), or (iii) an adjoint
action (Equation (\ref{Adjoint})). For simplicity, we mainly use the
implementation (i), so that $\Lambda _{a}(b)$ can be read row-wise in the
multiplicaction table
 of $\mathcal{G}$ (see Equation (\ref{Lambda(a)})), in
which case we write $D_{C,K}^{(\Lambda )}$ for $D_{C,K}^{(\Phi )}$. In case
(ii), we will write $D_{C,K}^{(R)}$.

\begin{example}
\label{Example Dist for Klein group}The only non-cyclic group of order 4 is
the Klein four-group $\mathcal{K}$, defined by the multiplication table

\begin{equation}
\begin{tabular}{|c||c|c|c|c|}
\hline
$\ast $ & $e$ & $a$ & $b$ & $c$ \\ \hline\hline
$e$ & $e$ & $a$ & $b$ & $c$ \\ \hline
$a$ & $a$ & $e$ & $c$ & $b$ \\ \hline
$b$ & $b$ & $c$ & $e$ & $a$ \\ \hline
$c$ & $c$ & $b$ & $a$ & $e$ \\ \hline
\end{tabular}
\label{TableK}
\end{equation}%
so that 

\begin{equation}
\begin{tabular}{|c||c|c|c|c|}
\hline
& $e$ & $a$ & $b$ & $c$ \\ \hline\hline
$\Lambda _{e}$ & $e$ & $a$ & $b$ & $c$ \\ \hline
$\Lambda _{a}$ & $a$ & $e$ & $c$ & $b$ \\ \hline
$\Lambda _{b}$ & $b$ & $c$ & $e$ & $a$ \\ \hline
$\Lambda _{c}$ & $c$ & $b$ & $a$ & $e$ \\ \hline
\end{tabular}
\label{TableK2}
\end{equation}%
According to Equations (\ref{DC max}) and (\ref{DK max}), $D_{C}^{(\Lambda
)}(r,s)\in \{0,1,2,3\}$ and $D_{K}^{(\Lambda )}(r,s)\in \{0,1,\ldots,6\}$. From
(\ref{TableK2}) it follows 
\begin{equation}
\begin{tabular}{|c||c|c|c|c|}
\hline
$D_{C}^{(\Lambda )}$ & $e$ & $a$ & $b$ & $c$ \\ \hline\hline
$e$ & $0$ & $2$ & $2$ & $2$ \\ \hline
$a$ & $2$ & $0$ & $2$ & $2$ \\ \hline
$b$ & $2$ & $2$ & $0$ & $2$ \\ \hline
$c$ & $2$ & $2$ & $2$ & $0$ \\ \hline
\end{tabular}%
\;\;\;\;\text{and\ \ \ \ }%
\begin{tabular}{|c||c|c|c|c|}
\hline
$D_{K}^{(\Lambda )}$ & $e$ & $a$ & $b$ & $c$ \\ \hline\hline
$e$ & $0$ & $2$ & $4$ & $6$ \\ \hline
$a$ & $2$ & $0$ & $6$ & $4$ \\ \hline
$b$ & $4$ & $6$ & $0$ & $2$ \\ \hline
$c$ & $6$ & $4$ & $2$ & $0$ \\ \hline
\end{tabular}
\label{F and Dist for K}
\end{equation}%
so the forbidden values of $D_{C}^{(\Lambda )}(r,s)$ are $\{1,3\}$ and the
forbidden values of $D_{K}^{(\Lambda )}(r,s)$ are $\{1,3,5\}$. Note that $%
\mathcal{K}$ is abelian (as any group whose cardinality is the square of a
prime number) since the multiplication table
 in Equation (\ref{TableK}) is
symmetric and every element other than the identity has order $2$, i.e.,
every element is its own inverse. Therefore, 
\begin{equation*}
R_{r}(s)=s\ast r^{-1}=s\ast r=r\ast s=\Lambda _{r}(s),
\end{equation*}%
i.e., the isomorphic copies $\Lambda _{r},R_{r}\in \mathrm{Sym}(\mathcal{K})$
are the same for all $r\in \mathcal{K}$, which implies $%
D_{C,K}^{(R)}=D_{C,K}^{(\Lambda )}$. Labeling the elements $e,a,b,c$ as $%
1,2,3,4$, one can locate the four copies $\{\Lambda _{r}:r\in \mathcal{K}\}$
of the group $\mathcal{K}$ in the Kendall adjacency graph of $\mathrm{Sym}(4)$,
Figure \ref{Figure5}, and read there the distances in the right table of Equation (%
\ref{F and Dist for K}). For example,%
\begin{equation*}
D_{K}^{(\Lambda )}(a,b)=d_{K}(\Lambda _{a},\Lambda
_{b})=d_{K}(aecb,bcea)=d_{K}(2143,3412)=6.
\end{equation*}
\end{example}

As a final remark, note that when $\mathcal{G}=\mathrm{Sym}(L)$, $%
D_{C,K}^{(\Phi )}(\mathbf{r},\mathbf{s})$ does not become $d_{C,K}(\mathbf{r}%
,\mathbf{s})$, as one might think. The reason is that, in that event, $%
d_{C,K}(\mathbf{r},\mathbf{s})$ is defined on $\mathrm{Sym}(L)\times \mathrm{%
Sym}(L)$, while $D_{C,K}^{(\Phi )}(\mathbf{r},\mathbf{s})=d_{C,K}(\Phi (%
\mathbf{r}),\Phi (\mathbf{s}))$, where $d_{C,K}(\Phi (\mathbf{r}),\Phi (%
\mathbf{s}))$ is defined on $\mathrm{Sym}(\mathrm{Sym}(L))\times \mathrm{Sym}%
(\mathrm{Sym}(L))=\mathrm{Sym}(L!)\times \mathrm{Sym}(L!)$. In other terms,
the definition domain and the range of Cayley's isomorphism $\Phi :$ $%
\mathrm{Sym}(L)\rightarrow \mathrm{Sym}(L!)$ are different also in the
particular case $\mathcal{G}=\mathrm{Sym}(L)$, which prevents $\Phi $ from
becoming the identity (unless $L=2$). However, this does not prevent $%
d_{C,K}(\mathbf{r},\mathbf{s})$ and $D_{C,K}^{(\Lambda )}(\mathbf{r},\mathbf{%
s})$ from providing the same qualitative and even quantitative information, as shown in Example %
\ref{ExampleConsistency} below and Section~\ref{sec:NumSimul}. This fact
supports the consistency of our approach to group metrics based on Cayley's
isomorphism.

\begin{example}
\label{ExampleConsistency} Tables (\ref{d_C(Lr,Ls) for Sym(6)}) and (\ref%
{d_K(Lr,Ls) for Sym(6)}) below show the distances $d_{C}(\Lambda _{\mathbf{r}%
},\Lambda _{\mathbf{s}})=:D_{C}^{(\Lambda )}(\mathbf{r},\mathbf{s})$ and $%
d_{K}(\Lambda _{\mathbf{r}},\Lambda _{\mathbf{s}})=:D_{K}^{(\Lambda )}(%
\mathbf{r},\mathbf{s})$ for $\mathbf{r},\mathbf{s}\in \mathrm{Sym}(3)$, and $%
\Lambda _{\mathbf{r}}=\Phi (\mathbf{r})$, $\Lambda _{\mathbf{s}}=\Phi (\mathbf{s%
})\in \mathrm{Sym}(6)$, see table (\ref{Sym(Sym(3))F}):

\begin{equation}
\begin{tabular}{|c||c|c|c|c|c|c|}
\hline
$d_{C}(\Lambda _{\mathbf{r}},\Lambda _{\mathbf{s}})$ & $\Lambda _{123}$ & $%
\Lambda _{132}$ & $\Lambda _{213}$ & $\Lambda _{231}$ & $\Lambda _{312}$ & $%
\Lambda _{321}$ \\ \hline\hline
$\Lambda _{123}$ & $0$ & $3$ & $3$ & $4$ & $4$ & $3$ \\ \hline
$\Lambda _{132}$ & $3$ & $0$ & $4$ & $3$ & $3$ & $4$ \\ \hline
$\Lambda _{213}$ & $3$ & $4$ & $0$ & $3$ & $3$ & $4$ \\ \hline
$\Lambda _{231}$ & $4$ & $3$ & $3$ & $0$ & $4$ & $3$ \\ \hline
$\Lambda _{312}$ & $4$ & $3$ & $3$ & $4$ & $0$ & $3$ \\ \hline
$\Lambda _{321}$ & $3$ & $4$ & $4$ & $3$ & $3$ & $0$ \\ \hline
\end{tabular}
\label{d_C(Lr,Ls) for Sym(6)}
\end{equation}%

\begin{equation}
\begin{tabular}{|c||c|c|c|c|c|c|}
\hline
$d_{K}(\Lambda _{\mathbf{r}},\Lambda _{\mathbf{s}})$ & $\Lambda _{123}$ & $%
\Lambda _{132}$ & $\Lambda _{213}$ & $\Lambda _{231}$ & $\Lambda _{312}$ & $%
\Lambda _{321}$ \\ \hline\hline
$\Lambda _{123}$ & $0$ & $5$ & $5$ & $10$ & $10$ & $15$ \\ \hline
$\Lambda _{132}$ & $5$ & $0$ & $10$ & $15$ & $5$ & $10$ \\ \hline
$\Lambda _{213}$ & $5$ & $10$ & $0$ & $5$ & $15$ & $10$ \\ \hline
$\Lambda _{231}$ & $10$ & $15$ & $5$ & $0$ & $10$ & $5$ \\ \hline
$\Lambda _{312}$ & $10$ & $5$ & $15$ & $10$ & $0$ & $5$ \\ \hline
$\Lambda _{321}$ & $15$ & $10$ & $10$ & $5$ & $5$ & $0$ \\ \hline
\end{tabular}
\label{d_K(Lr,Ls) for Sym(6)}
\end{equation}%
For instance, if we encode the permutations of $\mathrm{Sym}(3)$ as 
\begin{equation}
123=1,\;132=2,\;213=3,\;231=4,\;312=5,\;321=6,  \label{shortening}
\end{equation}%
then%
\begin{equation*}
D_{C}^{(\Lambda )}(213,321)=d_{C}(\Lambda _{213},\Lambda
_{321})=d_{C}(341265,654321)=4,
\end{equation*}%
while 
\begin{equation*}
D_{K}^{(\Lambda )}(213,321)=d_{K}(\Lambda _{213},\Lambda
_{321})=d_{K}(341265,654321)=10.
\end{equation*}%
Note that if we replace $3$ by $1$ and $4$ by $2$ in Equation (\ref{d_C(Lr,Ls)
for Sym(6)}) for $d_{C}(\Lambda _{\mathbf{r}},\Lambda _{\mathbf{s}})$, then
we obtain \mbox{Equation (\ref{d_C(Sym(3))})} for $d_{C}(\mathbf{r},\mathbf{s})$.
Furthermore, if we divide $d_{K}(\Lambda _{\mathbf{r}},\Lambda _{\mathbf{s}})
$ in Equation (\ref{d_K(Lr,Ls) for Sym(6)}) by $3$, then we obtain \mbox{Equation (\ref%
{d_K(Sym(3))})} for $d_{K}(\mathbf{r},\mathbf{s})$, i.e., 
\begin{equation}
D_{K}^{(\Lambda )}(\mathbf{r},\mathbf{s})=3d_{K}(\mathbf{r},\mathbf{s})
\label{D=3d}
\end{equation}%
for all $\mathbf{r},\mathbf{s}\in \mathrm{Sym}(3)$. We conclude that the
results obtained using $d_{C,K}(\mathbf{r},\mathbf{s})$ in $\mathcal{G}=%
\mathrm{Sym}(3)$ and $D_{C,K}^{(\Lambda )}(\mathbf{r},\mathbf{s})$ in $\Phi (%
\mathcal{G})\subset \mathrm{Sym}(6)$ are equivalent. According to Equations (\ref{d_C(Lr,Ls) for Sym(6)}) and (\ref{d_K(Lr,Ls) for Sym(6)}), the allowed distances for $D_{C}^{(\Lambda )}$ are $\{0,3,4\}$ out of $\{0,1,\ldots,5\}$, while the allowed distances for $D_{K}^{(\Lambda )}$ are $%
\{0,5,10,15\}=\{5k:0\leq k\leq 3=d_{K,\max }(3)\}$ out of $\{0,1,\ldots,15\}$.
\end{example}

\subsection{Distances with Respect to a Generating Set}

\label{sec:DistwithGenerators}

For the time being, let $\mathcal{G}$ be a finite or infinite group. A
finite set $S=\{s_{1},\ldots,s_{n}\}\subset \mathcal{G}$ is a \textit{%
generating set} (or generator) of $\mathcal{G}$ if every $a\in \mathcal{G}$
can be written as a finite product of elements of $S$ and their inverses. In
particular, groups generated by a single element are called cyclic. For
example, $\{\theta ^{0},\theta ^{1},\ldots,\theta ^{n-1}\}$ endowed with $%
\theta ^{i}\ast \theta ^{j}=\theta ^{k}$, where $k=i+j$ \textrm{mod} $n$ is
a cyclic group of order $n$ with generator $S=\{\theta ^{1}\}$. The (edit)
distance (or \textit{word metric}) $d_{S}(a,b)$ between the elements $a$ and 
$b$ of a finitely generated group (in particular of a finite group) $%
\mathcal{G}$ is defined as the minimum number of elements from the
generating set $S$ needed to transform $a$ into $b$. That is, if $b=a\ast
s_{1}\ast \ldots\ast s_{k}$, where $s_{i}\in S$ (or $s_{i}^{-1}\in S$), then $%
d_{S}(a,b)$ is the smallest possible value of $k$. Therefore, the distance $%
d_{S}$ depends on the generating set $S$. In particular, if $\mathcal{G}=$ $%
\mathrm{Sym}(L)$, then the Cayley distance $d_{C}(\mathbf{r},\mathbf{s})$ of
Section~\ref{sec:EditDist} is the distance $d_{S}$ with respect to the
generating set of all transpositions, while the Kendall distance $d_{K}(%
\mathbf{r},\mathbf{s})$ is the distance $d_{S}$ with respect to the
generating set of all adjacent transpositions.

\begin{example}
\label{Example Dist for Cyclic group}For the cyclic group $\mathcal{G}%
=\{\theta ^{0},\theta ^{1},\theta ^{2},\theta ^{3}\}$\ of Example\ \ref%
{ExampleCyclicGroup}, the distances with respect to the generating set $%
S=\{\theta ^{1}\}$ are the following:%
\begin{equation}
\begin{tabular}{|c||c|c|c|c|}
\hline
$d_{S}(\theta ^{i},\theta ^{j})$ & $\theta ^{0}$ & $\theta ^{1}$ & $\theta
^{2}$ & $\theta ^{3}$ \\ \hline\hline
$\theta ^{0}$ & $0$ & $1$ & $2$ & $3$ \\ \hline
$\theta ^{1}$ & $1$ & $0$ & $1$ & $2$ \\ \hline
$\theta ^{2}$ & $2$ & $1$ & $0$ & $1$ \\ \hline
$\theta ^{3}$ & $3$ & $2$ & $1$ & $0$ \\ \hline
\end{tabular}
\label{d_S Cyclic Group}
\end{equation}%
As for the distances $D_{K}^{(\Lambda )}(\theta ^{i},\theta
^{j})=d_{K}(\Lambda _{\theta ^{i}},\Lambda _{\theta ^{j}})$, we find (see
Equation (\ref{TableF})):%
\begin{equation}
\begin{tabular}{|c||c|c|c|c|}
\hline
$D_{K}^{(\Lambda )}(\theta ^{i},\theta ^{j})$ & $\theta ^{0}$ & $\theta ^{1}$
& $\theta ^{2}$ & $\theta ^{3}$ \\ \hline\hline
$\theta ^{0}$ & $0$ & $3$ & $4$ & $3$ \\ \hline
$\theta ^{1}$ & $3$ & $0$ & $3$ & $4$ \\ \hline
$\theta ^{2}$ & $4$ & $3$ & $0$ & $3$ \\ \hline
$\theta ^{3}$ & $3$ & $4$ & $3$ & $0$ \\ \hline
\end{tabular}
\label{dist_F Cyclic Group}
\end{equation}%
For example,%
\begin{equation}
D_{K}^{(\Lambda )}(\theta ^{2},\theta ^{3})=d_{K}(\Lambda _{\theta
^{2}},\Lambda _{\theta ^{3}})=d_{K}(\theta ^{2}\theta ^{3}\theta ^{0}\theta
^{1},\theta ^{3}\theta ^{0}\theta ^{1}\theta ^{2})=3.  \label{distF}
\end{equation}%
If right translations (\ref{Transcript}) are used instead of left
translations (\ref{FreeAction}), then $D_{K}^{(R)}(\theta ^{i},\theta
^{j})=d_{K}(R_{\theta ^{i}},R_{\theta ^{j}})$ happens to be the same as in
Equation (\ref{dist_F Cyclic Group}). For example,%
\begin{equation}
D_{K}^{(R)}(\theta ^{2},\theta ^{3})=d_{K}(R_{\theta ^{2}},R_{\theta
^{3}})=d_{K}(\theta ^{2}\theta ^{3}\theta ^{0}\theta ^{1},\theta ^{1}\theta
^{2}\theta ^{3}\theta ^{0})=3.  \label{distT}
\end{equation}%
If the group elements $\theta ^{0},\theta ^{1},\theta ^{2},\theta ^{3}$ are
labeled $1,2,3,4$, respectively, then the above distances can be read in the
Kendall adjacency graph of $\mathrm{Sym}(4)$, Figure \ref{Figure5}.
For example, distances (\ref{distF}) and (\ref{distT}) read $%
d_{K}(3412,4123) $ and $d_{K}(3412,2341)$, respectively.
\end{example}

\subsection{Discussion}

\label{sec:Discussion}

When comparing the distances $D_{C,K}^{(\Phi )}(a,b)$ and $d_{S}(a,b)$ for
finite groups, a possible advantage of the former is its expediency, in the
sense that $D_{C,K}^{(\Phi )}$ dispenses with generating sets and, hence,
with the search for minimal descriptions of $b$ as products of the form $%
a\ast s_{1}\ast \ldots\ast s_{k}$. In addition, there are algorithms (such as
the bubble-sort algorithm) that compute $D_{C}^{(\Phi )}$ in time $%
O(\left\vert \mathcal{G}\right\vert )$ and $D_{K}^{(\Phi )}$ in time $\
O(\left\vert \mathcal{G}\right\vert \log \left\vert \mathcal{G}\right\vert )$
\cite{Cicirello2019}. Computational issues are briefly discussed in Section %
\ref{sec:NumSimul}.

On the other hand, a possible shortcoming of the distances $D_{C,K}^{(\Phi
)} $ in applications is the existence of forbidden values pointed out in
Remark~\ref{RemarkGaps}. For instance, the presence of such gaps in the
distances between the algebraic representations of two coupled time series
(see Section~\ref{sec:DistforTS}) might be misinterpreted as a dynamical
characteristic of the underlying systems, e.g., full or generalized
synchronization. So, the forbidden values for $D_{C,K}^{(\Phi )}$ must be
identified in advance, which can be easily done by calculating the row $%
(D_{C,K}^{(\Phi )}(e,c):c\in \mathcal{G}))$ of the distance matrix (Remark %
\ref{RemarkGaps}). Alternatively, they can be identified using independent
white noises. We come back to this point in Section~\ref{sec:NumSimul}.

In sum, when embedding a group $\mathcal{G}$ in $\mathrm{Sym}(\left\vert 
\mathcal{G}\right\vert )$ via Cayley's isomorphism $\Phi $, we are encoding
the $\left\vert \mathcal{G}\right\vert $ elements $\{a_{1},\ldots,a_{\left\vert 
\mathcal{G}\right\vert }\}\in \mathcal{G}$ as the $\left\vert \mathcal{G}%
\right\vert $ permutations $\Phi (a)=(b_{1},b_{2},\ldots,b_{\left\vert \mathcal{%
G}\right\vert })$, where $(b_{1},b_{2},\ldots,b_{\left\vert \mathcal{G}%
\right\vert })$ is a shuffle of $(a_{1},\ldots,a_{\left\vert \mathcal{G}%
\right\vert })$; see Equation (\ref{Lambda(a)}) for $\Phi $ being the left
translation $a_{i}\mapsto \Lambda _{a_{i}}$. The penalty for doing so is a
more complex representation of the elements of $\mathcal{G}$. The pay-off is a general and
computationally efficient metric $D_{C,K}^{(\Phi )}$. In principle, there
may be symmetric groups $\mathrm{Sym}(M)$ with $M<\left\vert \mathcal{G}%
\right\vert $ in which $\mathcal{G}$ can be embedded, but finding such symmetric groups, in particular, the minimum-order one, is rather difficult in general~\cite{Johnson1971,Grechkoseeva2003}. In any case, note that in the practice of symbolic representation of time series, the alphabets used have low cardinality.


\section{Distances for Group-Valued Time Series and Algebraic Representations%
}

\label{sec:DistforTS}

In this section, we explore possible applications of permutation-based
distances to group-valued time series. Examples of group-valued time series
include binary and $n$-ary time series. In the first case, $\mathcal{G=}%
\{0,1\}$, endowed with the XOR operation (addition modulo 2); these time
series arise in digital communications and cryptography. The second example
is a generalization, also used in digital communications: $\mathcal{G=}%
\{0,1,\ldots,n-1\}$ endowed with addition modulo $n$.

The perhaps most familiar example of group-valued time series is the
ordinal representation of real-value time series, introduced in Section \ref%
{sec:OrdPat}. A generalization thereof is the concept of algebraic
representation.

\begin{definition}
We say that a symbolic representation $\alpha =(a_{t})_{t\geq 0}$ of a time
series is an algebraic time series if its elements $a_{t}$ belong to a
finite group $(\mathcal{G},\ast )$.
\end{definition}

Since here we are interested in practical applications, consider two finite $%
\mathcal{G}$-valued time series $\alpha =(a_{t})_{1\leq t\leq N}$ and $\beta
=(b_{t})_{1\leq t\leq N}$ of length $N$. In time series analysis, $\alpha $
and $\beta $ could be ordinal representations of two coupled real-valued
time series $(x_{t})_{1\leq t\leq N}$ and $(y_{t})_{1\leq t\leq N}$,
respectively. To carry out a data-driven analysis of the coupled dynamics of
the underlying systems (think of various types of synchronization), or to
measure the similarity between $\alpha $ and $\beta $, there are a number of
metrics that we review in Section~\ref{sec:Metrics}. In Section \ref%
{sec:ExtractInfo}, we discuss how to extract information with those metrics.

\subsection{String Metrics for Group-Valued Time Series}

\label{sec:Metrics}

Below, we mention perhaps the most common metrics. Each of them targets
specific situations.

\begin{itemize}
\item[\textbf{(i)}] Some of the metrics to quantify the similarity of two symbolic
time series such as $\alpha $ and $\beta $ are based on the probability
distributions of their symbols (estimated by their frequencies) \cite%
{Rachev2013}. This category includes the Kullback--Leibler (KL) divergence
(usually symmetrized via an arithmetic or harmonic mean)~\cite{Monetti2009},
the Jensen--Shannon (JS) divergence~\cite{Lin1991}, the JS distance (which is
the square root of the JS divergence)~\cite{Endres2003}, the permutation JS
distance~\cite{Zunino2022,Zunino2024}, the Hellinger distance \cite%
{Hellinger1909}, the Wasserstein distance~\cite{Kantorovich1942,Figalli2021}%
, the total variation distance~\cite{Pinsker1964,Bhattacharyya2023} and
more. Since in this paper we are interested in harnessing the algebraic
structure of the symbolic data (if any), we will dispense with entropic
distances.

\item[\textbf{(ii)}] One can also exploit the algebraic structure of $\mathcal{G}$
and calculate the transcription of $\alpha $ and $\beta $~\cite{Monetti2009}%
, that is, the time series $\tau =(\tau _{t})_{t\geq 0}$, where $\tau
_{t}=b_{t}\ast a_{t}^{-1}$ (right translations by $a_{t}$) or $\tau
_{t}=a_{t}^{-1}\ast b_{t}$ (left translations by $a_{t}^{-1}$), see
Equations (\ref{Transcript}) and (\ref{FreeAction}). Trancriptions of
coupled time series in an ordinal representation have been used to study
different aspects of coupled dynamics: complexity \cite%
{Monetti2009,Amigo2012}, synchronization~\cite{Monetti2009,Amigo2012},
information directionality (or causality)~\cite{Monetti2013B}, features for
classification~\cite{Pilarczyk2023}, etc. Interestingly, if $\mathcal{G}=%
\mathrm{Sym}(L)$, then the distance between the ordinal patterns $a_{t}$ and 
$b_{t}$ can be written as the norm $\left\Vert \cdot \right\Vert _{C,K}$ of
the transcript $a_{t}^{-1}\circ b_{t}$, see Equation (\ref{norm2}).
Otherwise, we embed $\mathcal{G}$ into $\mathrm{Sym}(\mathcal{G})$ via
Cayley's isomorphism $\Phi :\mathcal{G}\rightarrow \mathrm{Sym}(\mathcal{G})$
and, again, the distance between the ordinal patterns $\Phi (a_{t})$ and $%
\Phi (b_{t})$ can be written as the norm $\left\Vert \cdot \right\Vert
_{C,K} $ of the transcript $\Phi (a_{t})^{-1}\circ \Phi (b_{t})$, see
Equation (\ref{dist(a,b) 2}).

\item[\textbf{(iii)}] Since a window $a_{t}^{W}:=a_{t},a_{t+1},\ldots,a_{t+W-1}$ of size 
$W$ of any $\mathcal{G}$-valued time series $\alpha =(a_{t})_{t\geq 0}$ can
be viewed as a string of symbols of length $W$, we can borrow a number of
string metrics from information theory, computer science and computational
linguistics to compare $a_{t}^{W}$ and $%
b_{t}^{W}:=b_{t},b_{t+1},\ldots,b_{t+W-1} $, where (unlike permutations) these
strings can have repeated symbols. Thus, the \textit{Hamming distance}
between two strings of equal length is the number of positions at which the
corresponding symbols differ~\cite{Hamming1950}. The \textit{%
Damerau--Levenshtein distance} considers insertions, deletions,
substitutions and adjacent transpositions of symbols \cite%
{Damerau1964,Levenshtein1966,Cormen2009}. Such metrics are also examples of
edit distances. Finally, we also mention the \textit{Jaro--Winkler
similarity coefficient} (not a true distance) which, like the Hamming
distance, is based on symbol matching~\cite{Jaro1989,Winkler1990}.
\end{itemize}

\subsection{Extracting Information with $d_{C,K}$ and $D_{C,K}^{(\Phi )}$}

\label{sec:ExtractInfo}

Next, we focus on the distances $d_{C,K}$ for the group $\mathrm{Sym}(L)$
(Section~\ref{sec:EditDist}) and $D_{C,K}^{(\Phi )}$ for other groups
(Section~\ref{sec:EditDist}) and their applications to the analysis of $%
\mathcal{G}$-valued time series and algebraic representations. The idea is
to measure the distance between (A) simultaneous symbols $a_{t}$ and $b_{t}$%
, or (B) concurrent windows $a_{t}^{W}$ and $b_{t}^{W}$, and thereby
characterize the similarity or dissimilarity of the symbolic time series $%
\alpha $ and $\beta $. To this end, we consider sliding windows $a_{t}^{W}$
and $b_{t}^{W}$, $1\leq t\leq N-W+1$, with the same size $W\geq 1$, where we
allow $W=1$ in order to include distances between simultaneous symbols.

\begin{enumerate}[label=,leftmargin=4.3em,labelsep=8mm]
\item[\textbf{\mbox{CASE A:}}]  $W=1.$   To unify the notation, we will write $\mathrm{%
dist}(a_{t},b_{t})$ for the distance between the elements $a_{t},b_{t}\in 
\mathcal{G}$, with the understanding that $\mathrm{dist}%
(a_{t},b_{t})=d_{C,K}(a_{t},b_{t})$ if $\mathcal{G}=\mathrm{Sym}(L)$ and $%
\mathrm{dist}(a_{t},b_{t})=D_{C,K}^{(\Phi )}(a_{t},b_{t})$ otherwise.
Therefore, 
\begin{equation}
\mathrm{dist}(a_{t},b_{t})\in \{0,1,\ldots,\mathrm{dist}_{\max }\},
\label{Upper bounds}
\end{equation}%
where 
\begin{equation}
\mathrm{dist}_{\max }\left\{ 
\begin{tabular}{ll}
$=L-1$ & if $\mathrm{dist}=d_{C}$ (Equation (\ref{Range d_C})), \\ 
$=L(L-1)/2$ & if $\mathrm{dist}=d_{K}$ (Equation (\ref{Range d_K})), \\ 
$\leq \left\vert \mathcal{G}\right\vert -1$ & if $\mathrm{dist}=D_{C}^{(\Phi
)}$ (Equation (\ref{DC max})), \\ 
$\leq \left\vert \mathcal{G}\right\vert (\left\vert \mathcal{G}\right\vert
-1)/2$ & if $\mathrm{dist}=D_{K}^{(\Phi )}$ (Equation (\ref{DK max})).%
\end{tabular}%
\right.  \label{dist_max}
\end{equation}%
where the inequalities in Equation (\ref{dist_max}) allow for the
possibility that $D_{C,K,\max }^{(\Phi )}$ is a forbidden distance (Remark %
\ref{RemarkGaps}). As a result of calculating $\mathrm{dist}(a_{t},b_{t})$
for $1\leq t\leq N$, we obtain the integer-valued time series%
\begin{equation}
(\mathrm{dist}(a_{t},b_{t}))_{1\leq t\leq N}.  \label{dist TS}
\end{equation}%
According to Equation (\ref{dist_max}), $d_{C,\max }<d_{K,\max }$ and $%
D_{C,\max }^{(\Phi )}<D_{K,\max }^{(\Phi )}$, except for $L=2.$ Therefore, $%
d_{K}$ and $D_{K}^{(\Phi )}$ have greater differentiating power in
applications than their Cayley counterparts due to their larger ranges.

\item[\textbf{\mbox{CASE B:}}]  $W>1$. Consider now the windows $%
a_{t}^{W}=(a_{t},a_{t+1},\ldots,a_{t+W-1})$ and $%
b_{t}^{W}=(b_{t},b_{t+1},\ldots,b_{t+W-1})$ as $W$-dimensional vectors in the
corresponding Cartesian product of the metric space $(\mathcal{G},\mathrm{%
dist})$. In this case, we have the whole family of $l_{p}$ distances, $p\geq
1$, at our disposal. Well-known instances include the so-called \textit{%
Manhattan distance},%
\begin{equation}
\mathrm{dist}_{1}(a_{t}^{W},b_{t}^{W})=\sum_{k=0}^{W-1}\mathrm{dist}%
(a_{t+k},b_{t+k}),  \label{Manhattan dist}
\end{equation}%
the \textit{Euclidean distance},%
\begin{equation}
\mathrm{dist}_{2}(a_{t}^{W},b_{t}^{W})=\left( \sum_{k=0}^{W-1}\mathrm{dist}%
(a_{t+k},b_{t+k})^{2}\right) ^{1/2},  \label{Euclidean dist}
\end{equation}%
and the \textit{Chebychev distance},%
\begin{equation}
\mathrm{dist}_{\infty }(a_{t}^{W},b_{t}^{W})=\max \left\{ \mathrm{dist}%
(a_{t+k},b_{t+k}):0\leq k\leq W-1\right\} .  \label{Chebyshev dist}
\end{equation}%
As a result, we obtain the time series%
\begin{equation}
(\mathrm{dist}_{p}(a_{t}^{W},b_{t}^{W}))_{1\leq t\leq N-W+1}
\label{dist_p TS}
\end{equation}%
which is integer-valued for $=1,\infty $, and real-valued otherwise.

\end{enumerate}

Once the metric information from the $\mathcal{G}$-valued time series $%
\alpha $ and $\beta $ has been collected element-wise (\ref{dist TS})
and/or window-wise (\ref{dist_p TS}), one can proceed in several ways to
process the information. We discuss some simple ways in Section \ref%
{sec:NumSimul}.


\section{Numerical Simulations}

\label{sec:NumSimul}

In this section, we illustrate the application of permutation-based distances
to algebraic representations with numerical simulations. To this end, we
revisit a model composed of two unidirectionally coupled, non-identical
Henon systems, used in~\cite{Amigo2024} to study generalized
synchronization. The equations of the driver $X$ are%
\begin{equation}
\left\{ 
\begin{array}{l}
x_{t+1}^{(1)}=1.4-(x_{t}^{(1)}{})^{2}+0.1x_{t}^{(2)} \\ 
x_{t+1}^{(2)}=x_{t}^{(1)}%
\end{array}%
\right.  \label{HenonMapX}
\end{equation}%
and the equations of the responder $Y$ are 
\begin{equation}
\left\{ 
\begin{array}{l}
y_{t+1}^{(1)}=1.4-[Cx_{t}^{(1)}y_{t}^{(1)}+(1-C)(y_{t}^{(1)}{})^{2}]+0.3y_{t}^{(2)}
\\ 
y_{t+1}^{(2)}=y_{t}^{(1)}%
\end{array}%
\right.  \label{HenonMap2}
\end{equation}%
where $C\geq 0$ is the \textit{coupling strength}. It is numerically proved
in~\cite{Amigo2024} that this system has \textit{generalized synchronization%
} for $C$ in a small interval around $0.55$ and for $C\gtrsim 1$ (Figure 3 of \cite{Amigo2024}).

For a given coupling strength $C$, let $x=(x_{t}^{(1)})_{1\leq t\leq 10000}$
and $y=(y_{t}^{(1)})_{1\leq t\leq 10000}$ be two stationary time series of
length $N=10,000$ composed of the first components of the states $%
x_{t}=(x_{t}^{(1)},x_{t}^{(2)})$ of the driver and $%
y_{t}=(y_{t}^{(1)},y_{t}^{(2)})$ of the responder, respectively, and
generated with seeds $x_{0}=(0,0.9)$ and $y_{0}=(0.75,0)$ (after discarding
the initial transient). Let $\alpha =(\mathbf{r}_{t})_{1\leq t\leq
10000-L+1} $ and $\beta =(\mathbf{s}_{t})_{1\leq t\leq 10000-L+1}$ be the
algebraic representations of $x$ and $y$ with ordinal patterns of length $%
3\leq L\leq 6 $. The values chosen for the coupling strength are $%
C=0.30,0.55,1.10$.

Next we computed different types of distances between $\alpha $ and $\beta $
from those presented in Section~\ref{sec:ExtractInfo}. Here, we present only
the results with the Kendall distances $d_{K}(\mathbf{r}_{t},\mathbf{s}_{t})$
and $D_{K}^{(\Lambda )}(\mathbf{r}_{t},\mathbf{s}_{t})$ because, as
explained there, they have greater differentiating power than $d_{C}$ and $%
D_{C}^{(\Lambda )}$. As for the distances $\mathrm{dist}_{p}(\mathbf{r}%
_{t}^{W},\mathbf{s}_{t}^{W})$, we used $p=1,2,\infty $ (\mbox{Equations (\ref%
{Manhattan dist})--(\ref{Chebyshev dist})}). Irrational values of $\mathrm{dist%
}_{2}(\mathbf{r}_{t}^{W},\mathbf{s}_{t}^{W})$ were rounded to the integer $n$
if $\mathrm{dist}_{2}(\mathbf{r}_{t}^{W},\mathbf{s}_{t}^{W})\in
(n-0.5,\,n+0.5]$. To facilitate analysis, we transformed the data $(d_{K}(%
\mathbf{r}_{t},\mathbf{s}_{t}))_{1\leq t\leq N-L+1}$, $(D_{K}^{(\Lambda )}(%
\mathbf{r}_{t},\mathbf{s}_{t}))_{1\leq t\leq N-L+1}$ and $(\mathrm{dist}%
_{p}(\mathbf{r}_{t}^{W},\mathbf{s}_{t}^{W}))_{1\leq t\leq N-L-W+2}$ into (empirical)
probability distributions for the distance values.

Figure~\ref{Figure2} illustrates CASE A of Section~\ref{sec:ExtractInfo},
i.e., $W=1$. Here, $\mathcal{G}=\mathrm{Sym}(4)$ (top row) and $\mathcal{G}=%
\mathrm{Sym}(5)$ (bottom row). The main conclusions can be summarized as
follows.

\begin{itemize}
\item For $C=0.30$ (no synchronization, panels (a) and (d)), all possible
values $\{0,1,\ldots,L(L-1)/2\}$ of $d_{K}$ are realized.

\item For $C=0.55$ (``weak synchronization'', panels (b) and (e)), only the greater values of $d_{K}$ are allowed.

\item For $C=1.10$ (``strong
synchronization'', panels (c) and (f)), only the smaller
values of $d_{K}$ are allowed.

\item So, $d_{K}$ detects that the generalized synchronizations at $C=0.55$
and $C=1.10$ are different: the former forbids the shorter distances between
simultaneous ordinal patterns $\mathbf{r}_{t}$ and $\mathbf{s}_{t}$, while
the latter forbids large distances.

\item The results for each $C$ are consistently similar.
\end{itemize}

We conclude that the distance $d_{K}$ is sensitive to dynamical changes in
coupled systems and robust with respect to the length of the ordinal
patterns.
\begin{figure}[H]
\centering
\includegraphics[scale=0.58]{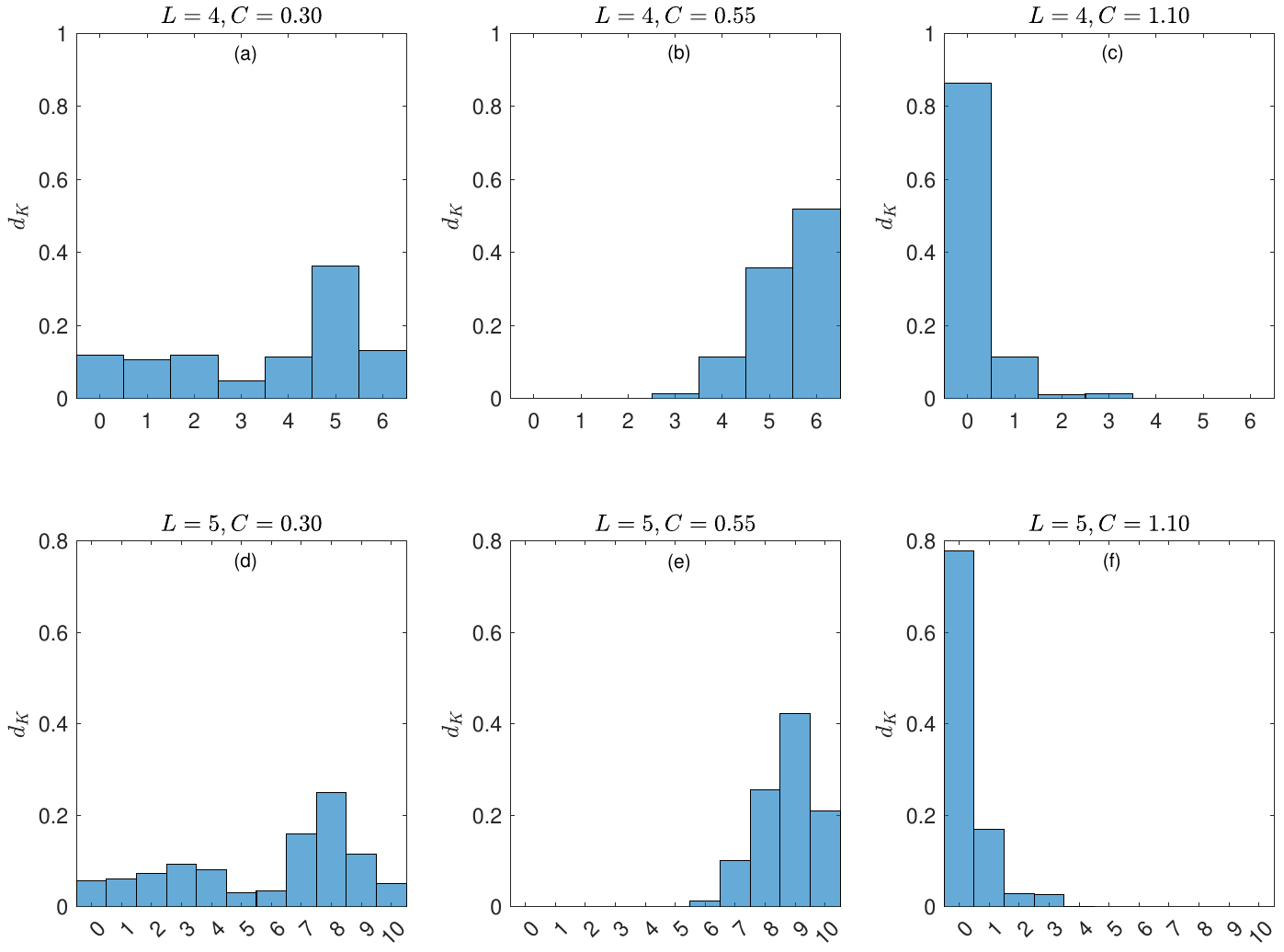}
\caption{\textbf{Top row}: Probability distributions of the Kendall distances $d_{K}(\mathbf{r}_{t},\mathbf{s}_{t})$ for the algebraic representation of the time series $x$ and $y$ with the group $\mathcal{G}=\mathrm{Sym}(4)$ (i.e., ordinal patterns of length $L=4$) and coupling strengths $C=0.30$ (\textbf{left panel}), $0.55$ (\textbf{middle panel}) and $1.10$ (\textbf{right panel}). \textbf{Bottom row}: Same as top row for the representation group $\mathcal{G}=\mathrm{Sym}(5)$ (i.e., ordinal patterns of length $L=5$).}
\label{Figure2}
\end{figure}
At this point, we draw on Figure~\ref{Figure2} to, as in Example \ref%
{ExampleConsistency}, check the consistency of the results obtained with $%
d_{K}$ and $D_{K}^{(\Lambda )}$, this time using $\mathcal{G}=\mathrm{Sym}(4)
$ and $\mathcal{G}=\mathrm{Sym}(5)$. Figure~\ref{Figure3} shows the
probability distribution of the allowed distances for $D_{K}^{(\Lambda )}(%
\mathbf{r}_{t},\mathbf{s}_{t})$ with $L=4$ (panel (a)) and $L=5$ (panel (b)). The
coupling strength in both panels is $C=0.30$, so that all $L(L-1)/2+1$
allowed distances are realized. The allowed distances for $D_{K}^{(\Lambda
)}(\mathbf{r}_{t},\mathbf{s}_{t})$, listed along the horizontal axes in
Figure~\ref{Figure3}, happen to be $\{46k:0\leq k\leq 6=d_{K,\max }(4)\}$
for $L=4$ and $\{714k:0\leq k\leq 10=d_{K,\max }(5)\}$ for $L=5$. Comparison
of panels (a) and (b) of Figure~\ref{Figure3} with panels (a) and (d) of
Figure~\ref{Figure2}, respectively, shows that the probability distributions
of $D_{K}^{(\Lambda )}(\mathbf{r}_{t},\mathbf{s}_{t})$ and $d_{K}(\mathbf{r}%
_{t},\mathbf{s}_{t})$ are exactly the same for $L=4,5$ and $C=0.30$, except
for the labeling of the distances; notice the change of scales. {In fact, and similarly to Equation (\ref{D=3d}), numerical calculations show
that (i)%
\begin{equation}
D_{K}^{(\Lambda )}(\mathbf{r},\mathbf{s})=46d_{K}(\mathbf{r},\mathbf{s})
\label{D=46d}
\end{equation}%
for all $\mathbf{r},\mathbf{s\in }\mathrm{Sym}(4)$, where $46=\min
\{d_{K}(\Lambda _{\mathbf{r}},\Lambda _{\mathbf{s}})>0:\mathbf{r},\mathbf{%
s\in }\mathrm{Sym}(4)\}$, and (ii)%
\begin{equation}
D_{K}^{(\Lambda )}(\mathbf{r},\mathbf{s})=714d_{K}(\mathbf{r},\mathbf{s})
\label{D=714d}
\end{equation}%
for all $\mathbf{r},\mathbf{s\in }\mathrm{Sym}(5)$, where $714=\min
\{d_{K}(\Lambda _{\mathbf{r}},\Lambda _{\mathbf{s}})>0:\mathbf{r},\mathbf{%
s\in }\mathrm{Sym}(5)\}$. For example, $d_{K}(\Lambda _{1234},\Lambda
_{1243})=46$ and $d_{K}(\Lambda _{12345},\Lambda _{12354})=714$.} The same
occurs for $C=0.55$ and $C=1.10$ (not shown).

\begin{figure}[H]

\centering
\includegraphics[scale=0.4]{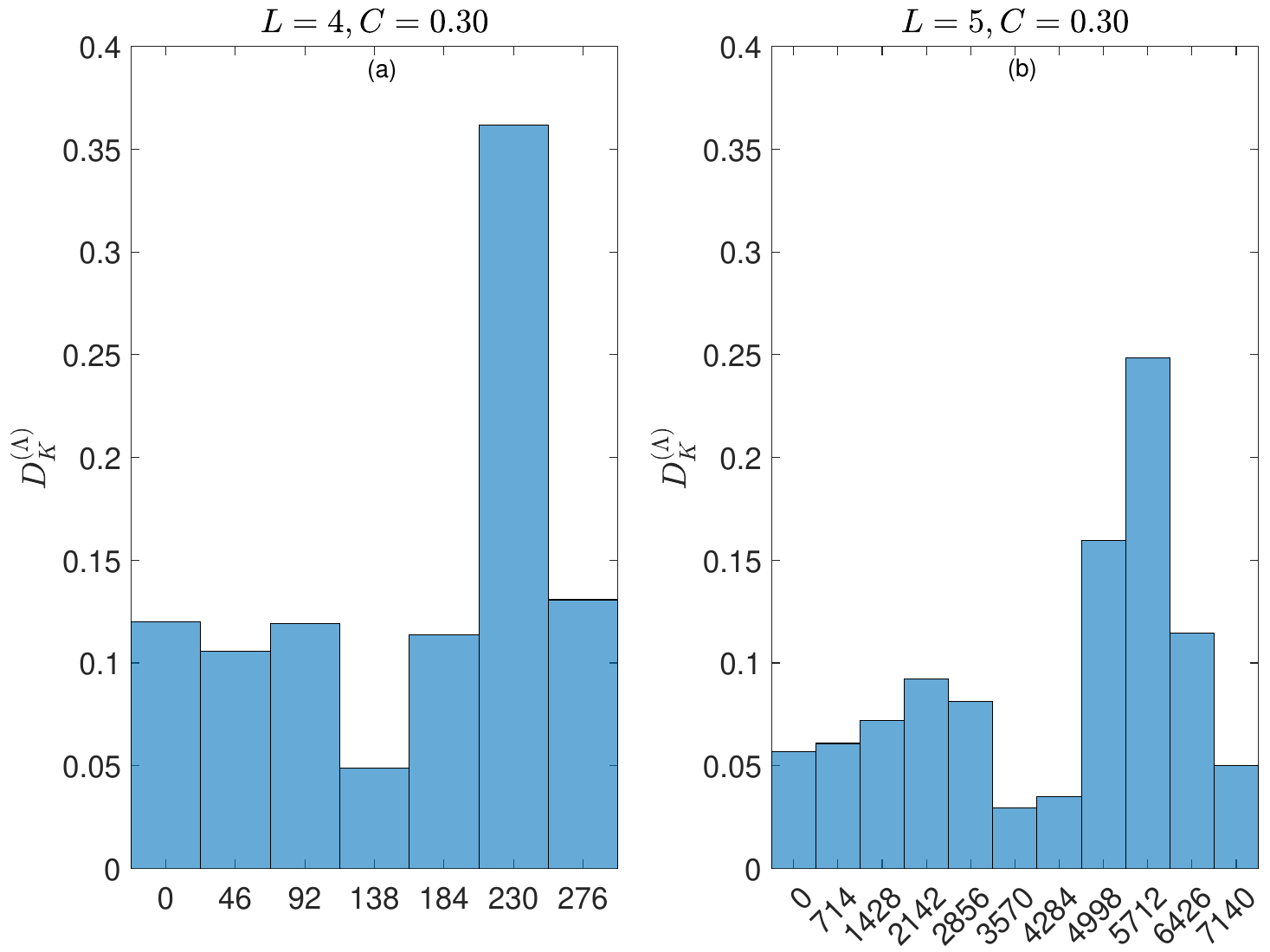}
\caption{Probability distributions of the allowed distances for $%
D_{K}^{(\Lambda )}(\mathbf{r}_{t},\mathbf{s}_{t})$ for the algebraic
representation of the time series $x$ and $y$ with the group $\mathcal{G}=%
\mathrm{Sym}(4)$ (panel (\textbf{a})) and $\mathcal{G}=\mathrm{Sym}(5)$ (panel (\textbf{b})). $%
C=0.30$ in both panels so that all $L(L-1)/2+1$ allowed distances for $%
D_{K}^{(\Lambda )}(\mathbf{r}_{t},\mathbf{s}_{t})$ (listed along the
horizontal axes) are actually realized.}
\label{Figure3}
\end{figure}

Figure~\ref{Figure4} illustrates CASE B of Section~\ref{sec:ExtractInfo},
i.e., $W>1$. Here, $W=4$ with $\mathcal{G}=\mathrm{Sym}(3)$, $\mathrm{dist}(%
\mathbf{r}_{t},\mathbf{s}_{t})=d_{K}(\mathbf{r}_{t},\mathbf{s}_{t})$, and
the distance $\mathrm{dist}_{p}(\mathbf{r}_{t}^{4},\mathbf{s}_{t}^{4})$ is
(i) $\mathrm{dist}_{1}(\mathbf{r}_{t}^{4},\mathbf{s}_{t}^{4})\in
\{0,1,\ldots,12\}$ in the top row, (ii) $\mathrm{dist}_{2}(\mathbf{r}_{t}^{4},%
\mathbf{s}_{t}^{4})\in \lbrack 0,6]$ in the middle row and (iii) $\mathrm{%
dist}_{\infty }(\mathbf{r}_{t}^{4},\mathbf{s}_{t}^{4})\in \{0,1,2,3\}$ in
the bottom row. The main conclusions can be summarized as follows.

\begin{itemize}
\item Due to the monotony property of the $p$-norms ($\left\Vert \cdot
\right\Vert _{p}\geq \left\Vert \cdot \right\Vert _{p^{\prime }}$ for $1\leq
p\leq $ $p^{\prime }\leq \infty $), the distances with smaller parameters $p$
($\mathrm{dist}_{1}$ and $\mathrm{dist}_{2}$ in Figure \ref{Figure4}) have greater
differentiating power.

\item The results shown in Figure \ref{Figure4} (obtained with sliding windows of 4
consecutive ordinal patterns of length 3) are qualitatively similar to the
results shown in Figure \ref{Figure2} (obtained with simultaneous pairs $(\mathbf{r}_{t},%
\mathbf{s}_{t})$ of ordinal patterns of lengths 4 and 5).
\end{itemize}

We conclude that the distances $\mathrm{dist}_{p}$ are also sensitive to
dynamical changes in coupled systems and robust with respect to the
parameter $p\geq 1$.

To wrap up the previous discussion, we are also going to compare the
computational times of $D_{C,K}^{(\Lambda )}(a_{t},b_{t})$ (Section \ref%
{sec:DistforGroupsSub1}) and $d_{S}(a_{t},b_{t})$ (Section \ref%
{sec:DistwithGenerators}), where $(a_{t})_{1\leq t\leq N}$, $(b_{t})_{1\leq
t\leq N}$ are $\mathcal{G}$-valued time series. Rather than using ad hoc
groups and coupled time series, we take advange of the above ordinal
representations $\alpha $ and $\beta $, and benchmark the computational cost
of computing $D_{C,K}^{(\Lambda )}(\mathbf{r}_{t},\mathbf{s}_{t})$ for $%
\mathcal{G}=\mathrm{Sym}(L)$, $3\leq L\leq 6$ (the usual ordinal pattern lengths in
applications), $N=$ 10,000 and $C=0.30$, against the computational cost of
calculating $d_{K}(\mathbf{r}_{t},\mathbf{s}_{t})$ for the same group and
settings. We choose $C=0.30$ so that all allowed ordinal patterns are
realized (see Figures~\ref{Figure2} and~\ref{Figure3}). Table \ref%
{TableTimes} shows the times in seconds of the corresponding calculations
with a laptop (Intel I9 processor, 8 cores, 64 GB of RAM, 8 GB  of GPU
memory) and a non-paralellized algorithm.

Altogether, the above numerical results support the usefulness of distances $%
d_{C,K}$, $D_{C,K}^{(\Phi )}$ and $\mathrm{dist}_{p}$ in the analysis of
group-valued time series.

\begin{table}[H]
\centering
\caption{Computation time in seconds of $d_{K}(\mathbf{r}_{t},\mathbf{s}_{t})$ and $D_{C,K}^{(\Lambda )}(\mathbf{r}_{t},\mathbf{s}_{t})$, $1\leq t\leq 10,000$.}\label{TableTimes}
\begin{tabular*}{\linewidth}{@{\extracolsep{\fill}} cccc}
\toprule
\boldmath{$\mathcal{G}$} & \boldmath{$\Phi (\mathcal{G})$} & \boldmath{$d_{K}(\mathbf{r}_{t},\mathbf{s}_{t})$} & \boldmath{$D_{C,K}^{(\Lambda )}(\mathbf{r}_{t},\mathbf{s}_{t})$}\\
\midrule
$\mathrm{Sym}(3)$ & $\mathrm{Sym}(6)$ & $0.009$ s & $0.022$ s \\ 
$\mathrm{Sym}(4)$ & $\mathrm{Sym}(24)$ & $0.010$ s & $0.038$ s \\ 
$\mathrm{Sym}(5)$ & $\mathrm{Sym}(120)$ & $0.011$ s & $0.143$ s \\ 
$\mathrm{Sym}(6)$ & $\mathrm{Sym}(720)$ & $0.012$ s & $2.837$ s \\  
\bottomrule
\end{tabular*}
\end{table}

\begin{figure}[H]
\centering
\includegraphics[scale=0.70]{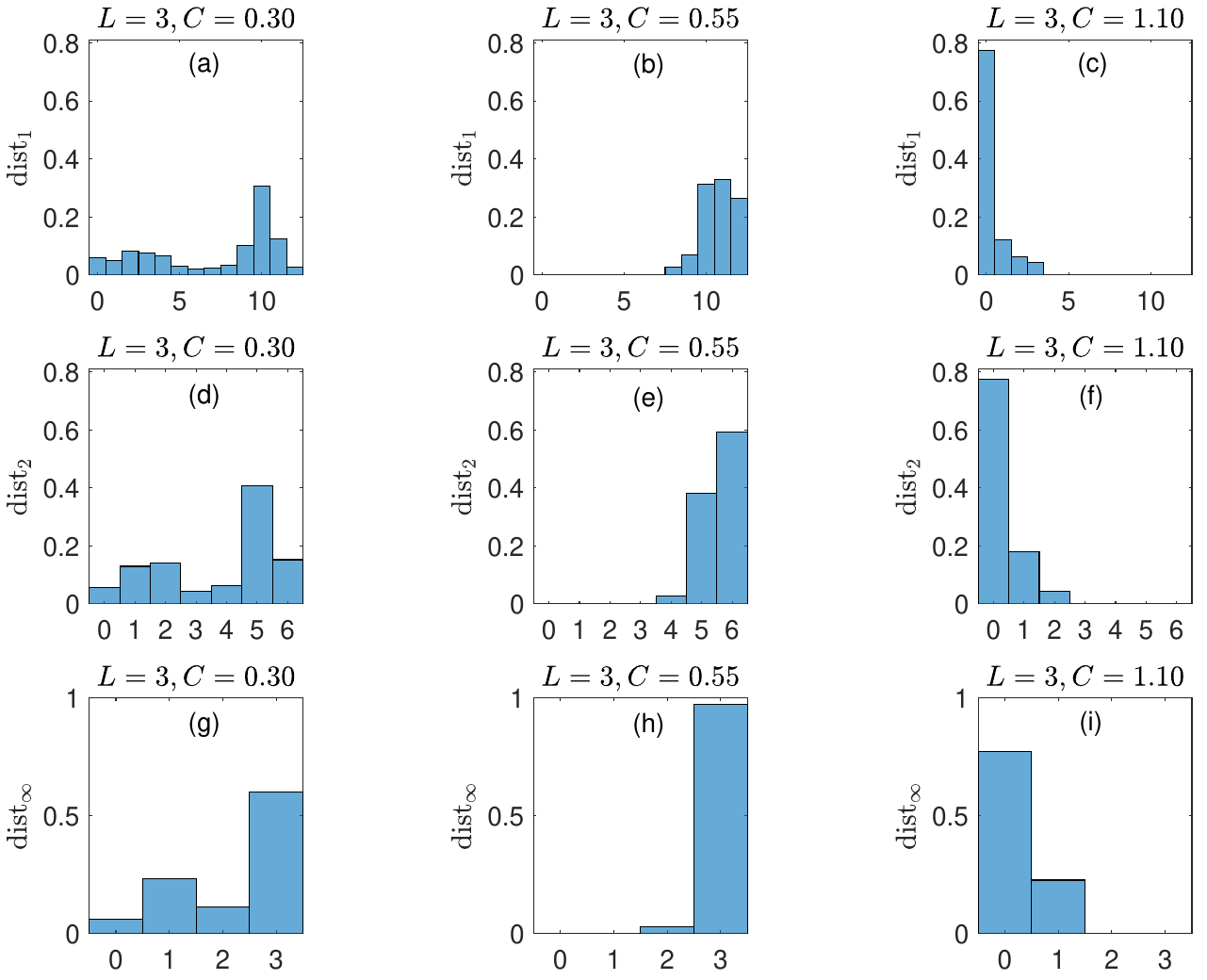}
\caption{\textbf{Top row}: Probability distributions of the distance $\mathrm{dist}%
_{1}(\mathbf{r}_{t}^{4},\mathbf{s}_{t}^{4})$ for the algebraic
representation of the time series $x$ and $y$ with the group $\mathcal{G}=%
\mathrm{Sym}(3)$ (ordinal patterns of length $L=3$) and coupling strengths $%
C=0.30$ (\textbf{left panel}), $0.55$ (\textbf{middle panel}) and $1.10$ (\textbf{right panel}). \textbf{Middle
row}: Same as top row for the distance $\mathrm{dist}_{2}(\mathbf{r}_{t}^{4},%
\mathbf{s}_{t}^{4})$. \textbf{Bottom row}: Same as top row for the distance $\mathrm{%
dist}_{\infty }(\mathbf{r}_{t}^{4},\mathbf{s}_{t}^{4})$.}
\label{Figure4}
\end{figure}


\section{Conclusions}

\label{sec:Conclusion}

The results presented in this paper are an outgrowth of the study of
transcripts and their applications to time series analysis in algebraic representations (Section~\ref{sec:DistforTS}), which are a generalization of transcripts in ordinal representations~\cite{Monetti2009}. Indeed, the concept of transcript from a group element $a\in \mathcal{G}$ to another $b\in \mathcal{G}$ or, for that matter, the right translation of $b$ by $a$
(Equation (\ref{Transcript})) leads directly to the isomorphism $\Phi
:   a\mapsto R(a,\cdot )=:R_{a}$ from $\mathcal{G}$ to a subgroup of the
symmetric group $\mathrm{Sym}(\mathcal{G})$ (Cayley's theorem). In turn, the elements of $\mathrm{Sym}(\mathcal{G})$ can be written as numerical or symbolic strings, which allows us to endow $\mathrm{Sym}(\mathcal{G})$ with any convenient edit distance, e.g., the Cayley distance $d_{C}$ or the Kendall distance $d_{K}$ of Section~\ref{sec:OrdPatyEditDist}. This being the case, the isomorphism $\Phi $ can be used to transport the distance $d_{C,K}$ in $%
\mathrm{Sym}(\mathcal{G})$ to $\mathcal{G}$, as we did in Section~\ref{sec:DistforGroups}. The result is the ordinal pattern-based distance for groups $D_{C,K}^{(\Phi )}(a,b)$ proposed in Definition~\ref{DefDist(a,b)}.

Metric properties of finite groups is an unsual tool in time series analysis
in algebraic representations. Even in the ordinal representation, distances
or similarities between time series are usually measured with functionals of
probability distributions such as divergences or functions thereof. There
are also distances defined in the groups themselves, based on generating
sets, which were the subject of Section~\ref{sec:DistwithGenerators}.
Actually, the distances $d_{C}$ and $d_{K}$ in the permutations groups,
discussed in Section~\ref{sec:OrdPatyEditDist}, are examples of distances
with respect to generating sets. A possible advantage of the ordinal
pattern-based distance proposed in this paper for any group $\mathcal{G}$ is
its simplicity and generality, since it dispenses with generating sets and
minimal descriptions of elements via generators. Furthermore, there are
general-purpose algorithms to efficiently calculate the distances $d_{C}$
and $d_{K}$ in $\mathrm{Sym}(\mathcal{G})$ for the low and moderate group
cardinalities used in practice, see Table~\ref{TableTimes}.

In the previous sections we have presented the mathematical underpinnings of our approach, which include group actions, Cayley's theorem, and group representations, as well as its practical implementation. It is remarkable that Cayley's theorem gives permutations (or ordinal patterns) a certain universality in algebraic representations of time series, although other
choices or isomorphisms can be more convenient in practice. For example, the Klein group (Example~\ref{Example Dist for Klein group}) is isomorphic to $\mathbb{Z}_{2}\times \mathbb{Z}_{2}$ endowed with XOR addition and the cyclic group $\{\theta ^{0},\theta ^{1},\ldots,\theta ^{n-1}\}$ endowed with $\theta ^{i}\ast \theta ^{j}=\theta ^{k}$, where $k=i+j$ \textrm{mod} $n$, is
isomorphic to $\{0,1,\ldots,n-1\}$ equipped with addition modulo $n$. Some of these groups were used in the previous sections to illustrate the theory. In contrast to the specificities of each group, the group distance introduced in Definition~\ref{DefDist(a,b)} is completely general, since the only input it needs is the multiplication table of the group, and can be efficiently computed. Possible applications were only touched upon in
Section~\ref{sec:DistforTS} because they are the subject of ongoing
research. The numerical simulations in Section~\ref{sec:NumSimul} show the potential of the metric tools discussed in this paper in the analysis of group-valued \mbox{time series}.

\clearpage
\appendix

\begin{landscapeblock}

\section[\appendixname~\thesection]{Cayley and Kendall Distances for the Group \boldmath{$\mathrm{Sym}(4)$} (Example~\ref{ExampleSym(3)}) \label{appendixa}}


\begin{table}[H]
\caption{Distance $d_{C}$ for $\mathrm{Sym}(4)$.}  \label{tab:table01}
\centering
\scriptsize
\setlength{\tabcolsep}{3pt}
\begin{tabular*}{\linewidth}{@{\extracolsep{\fill}} l *{24}{c} @{}}
\toprule
\boldmath{$d_{C}$} & \textbf{1234} & \textbf{1243} & \textbf{1324} & \textbf{1342} & \textbf{1423} & \textbf{1432} & \textbf{2134} & \textbf{2143} & \textbf{2314} & \textbf{2341} & \textbf{2413} & \textbf{2431} & \textbf{3124} & \textbf{3142} & \textbf{3214} & \textbf{3241} & \textbf{3412} & \textbf{3421} & \textbf{4123} & \textbf{4132} & \textbf{4213} & \textbf{4231} & \textbf{4312} & \textbf{4321} \\ \midrule
1234         & 0    & 1    & 1    & 2    & 2    & 1    & 1    & 2    & 2    & 3    & 3    & 2    & 2    & 3    & 1    & 2    & 2    & 3    & 3    & 2    & 2    & 1    & 3    & 2    \\
1243         & 1    & 0    & 2    & 1    & 1    & 2    & 2    & 1    & 3    & 2    & 2    & 3    & 3    & 2    & 2    & 1    & 3    & 2    & 2    & 3    & 1    & 2    & 2    & 3    \\
1324         & 1    & 2    & 0    & 1    & 1    & 2    & 2    & 3    & 1    & 2    & 2    & 3    & 1    & 2    & 2    & 3    & 3    & 2    & 2    & 3    & 3    & 2    & 2    & 1    \\
1342         & 2    & 1    & 1    & 0    & 2    & 1    & 3    & 2    & 2    & 1    & 3    & 2    & 2    & 1    & 3    & 2    & 2    & 3    & 3    & 2    & 2    & 3    & 1    & 2    \\
1423         & 2    & 1    & 1    & 2    & 0    & 1    & 3    & 2    & 2    & 3    & 1    & 2    & 2    & 3    & 3    & 2    & 2    & 1    & 1    & 2    & 2    & 3    & 3    & 2    \\
1432         & 1    & 2    & 2    & 1    & 1    & 0    & 2    & 3    & 3    & 2    & 2    & 1    & 3    & 2    & 2    & 3    & 1    & 2    & 2    & 1    & 3    & 2    & 2    & 3    \\
2134         & 1    & 2    & 2    & 3    & 3    & 2    & 0    & 1    & 1    & 2    & 2    & 1    & 1    & 2    & 2    & 3    & 3    & 2    & 2    & 1    & 3    & 2    & 2    & 3    \\
2143         & 2    & 1    & 3    & 2    & 2    & 3    & 1    & 0    & 2    & 1    & 1    & 2    & 2    & 1    & 3    & 2    & 2    & 3    & 1    & 2    & 2    & 3    & 3    & 2    \\
2314         & 2    & 3    & 1    & 2    & 2    & 3    & 1    & 2    & 0    & 1    & 1    & 2    & 2    & 3    & 1    & 2    & 2    & 3    & 3    & 2    & 2    & 3    & 1    & 2    \\
2341         & 3    & 2    & 2    & 1    & 3    & 2    & 2    & 1    & 1    & 0    & 2    & 1    & 3    & 2    & 2    & 1    & 3    & 2    & 2    & 3    & 3    & 2    & 2    & 1    \\
2413         & 3    & 2    & 2    & 3    & 1    & 2    & 2    & 1    & 1    & 2    & 0    & 1    & 3    & 2    & 2    & 3    & 1    & 2    & 2    & 3    & 1    & 2    & 2    & 3    \\
2431         & 2    & 3    & 3    & 2    & 2    & 1    & 1    & 2    & 2    & 1    & 1    & 0    & 2    & 3    & 3    & 2    & 2    & 1    & 3    & 2    & 2    & 1    & 3    & 2    \\
3124         & 2    & 3    & 1    & 2    & 2    & 3    & 1    & 2    & 2    & 3    & 3    & 2    & 0    & 1    & 1    & 2    & 2    & 1    & 1    & 2    & 2    & 3    & 3    & 2    \\
3142         & 3    & 2    & 2    & 1    & 3    & 2    & 2    & 1    & 3    & 2    & 2    & 3    & 1    & 0    & 2    & 1    & 1    & 2    & 2    & 1    & 3    & 2    & 2    & 3    \\
3214         & 1    & 2    & 2    & 3    & 3    & 2    & 2    & 3    & 1    & 2    & 2    & 3    & 1    & 2    & 0    & 1    & 1    & 2    & 2    & 3    & 1    & 2    & 2    & 3    \\
3241         & 2    & 1    & 3    & 2    & 2    & 3    & 3    & 2    & 2    & 1    & 3    & 2    & 2    & 1    & 1    & 0    & 2    & 1    & 3    & 2    & 2    & 1    & 3    & 2    \\
3412         & 2    & 3    & 3    & 2    & 2    & 1    & 3    & 2    & 2    & 3    & 1    & 2    & 2    & 1    & 1    & 2    & 0    & 1    & 3    & 2    & 2    & 3    & 1    & 2    \\
3421         & 3    & 2    & 2    & 3    & 1    & 2    & 2    & 3    & 3    & 2    & 2    & 1    & 1    & 2    & 2    & 1    & 1    & 0    & 2    & 3    & 3    & 2    & 2    & 1    \\
4123         & 3    & 2    & 2    & 3    & 1    & 2    & 2    & 1    & 3    & 2    & 2    & 3    & 1    & 2    & 2    & 3    & 3    & 2    & 0    & 1    & 1    & 2    & 2    & 1    \\
4132         & 2    & 3    & 3    & 2    & 2    & 1    & 1    & 2    & 2    & 3    & 3    & 2    & 2    & 1    & 3    & 2    & 2    & 3    & 1    & 0    & 2    & 1    & 1    & 2    \\
4213         & 2    & 1    & 3    & 2    & 2    & 3    & 3    & 2    & 2    & 3    & 1    & 2    & 2    & 3    & 1    & 2    & 2    & 3    & 1    & 2    & 0    & 1    & 1    & 2    \\
4231         & 1    & 2    & 2    & 3    & 3    & 2    & 2    & 3    & 3    & 2    & 2    & 1    & 3    & 2    & 2    & 1    & 3    & 2    & 2    & 1    & 1    & 0    & 2    & 1    \\
4312         & 3    & 2    & 2    & 1    & 3    & 2    & 2    & 3    & 1    & 2    & 2    & 3    & 3    & 2    & 2    & 3    & 1    & 2    & 2    & 1    & 1    & 2    & 0    & 1    \\
4321         & 2    & 3    & 1    & 2    & 2    & 3    & 3    & 2    & 2    & 1    & 3    & 2    & 2    & 3    & 3    & 2    & 2    & 1    & 1    & 2    & 2    & 1    & 1    & 0   \\
\bottomrule
\end{tabular*}
\end{table}

\begin{table}[H]
\caption{Distance $d_{K}$ for $\mathrm{Sym}(4)$.}\label{tab:table02}
\centering
\scriptsize
\setlength{\tabcolsep}{3pt}
\begin{tabular*}{\linewidth}{@{\extracolsep{\fill}} l *{24}{c} @{}}
\toprule
\boldmath{$d_{K}$} & \textbf{1234} & \textbf{1243} & \textbf{1324} & \textbf{1342} & \textbf{1423} & \textbf{1432} & \textbf{2134} & \textbf{2143} & \textbf{2314} & \textbf{2341} & \textbf{2413} & \textbf{2431} & \textbf{3124} & \textbf{3142} & \textbf{3214} & \textbf{3241} & \textbf{3412} & \textbf{3421} & \textbf{4123} & \textbf{4132} & \textbf{4213} & \textbf{4231} & \textbf{4312} & \textbf{4321} \\
\midrule
1234    & 0    & 1    & 1    & 2    & 2    & 3    & 1    & 2    & 2    & 3    & 3    & 4    & 2    & 3    & 3    & 4    & 4    & 5    & 3    & 4    & 4    & 5    & 5    & 6    \\
1243    & 1    & 0    & 2    & 3    & 1    & 2    & 2    & 1    & 3    & 4    & 2    & 3    & 3    & 4    & 4    & 5    & 5    & 6    & 2    & 3    & 3    & 4    & 4    & 5    \\
1324    & 1    & 2    & 0    & 1    & 3    & 2    & 2    & 3    & 3    & 4    & 4    & 5    & 1    & 2    & 2    & 3    & 3    & 4    & 4    & 3    & 5    & 6    & 4    & 5    \\
1342    & 2    & 3    & 1    & 0    & 2    & 1    & 3    & 4    & 4    & 5    & 5    & 6    & 2    & 1    & 3    & 4    & 2    & 3    & 3    & 2    & 4    & 5    & 3    & 4    \\
1423    & 2    & 1    & 3    & 2    & 0    & 1    & 3    & 2    & 4    & 5    & 3    & 4    & 4    & 3    & 5    & 6    & 4    & 5    & 1    & 2    & 2    & 3    & 3    & 4    \\
1432    & 3    & 2    & 2    & 1    & 1    & 0    & 4    & 3    & 5    & 6    & 4    & 5    & 3    & 2    & 4    & 5    & 3    & 4    & 2    & 1    & 3    & 4    & 2    & 3    \\
2134    & 1    & 2    & 2    & 3    & 3    & 4    & 0    & 1    & 1    & 2    & 2    & 3    & 3    & 4    & 2    & 3    & 5    & 4    & 4    & 5    & 3    & 4    & 6    & 5    \\
2143    & 2    & 1    & 3    & 4    & 2    & 3    & 1    & 0    & 2    & 3    & 1    & 2    & 4    & 5    & 3    & 4    & 6    & 5    & 3    & 4    & 2    & 3    & 5    & 4    \\
2314    & 2    & 3    & 3    & 4    & 4    & 5    & 1    & 2    & 0    & 1    & 3    & 2    & 2    & 3    & 1    & 2    & 4    & 3    & 5    & 6    & 4    & 3    & 5    & 4    \\
2341    & 3    & 4    & 4    & 5    & 5    & 6    & 2    & 3    & 1    & 0    & 2    & 1    & 3    & 4    & 2    & 1    & 3    & 2    & 4    & 5    & 3    & 2    & 4    & 3    \\
2413    & 3    & 2    & 4    & 5    & 3    & 4    & 2    & 1    & 3    & 2    & 0    & 1    & 5    & 6    & 4    & 3    & 5    & 4    & 2    & 3    & 1    & 2    & 4    & 3    \\
2431    & 4    & 3    & 5    & 6    & 4    & 5    & 3    & 2    & 2    & 1    & 1    & 0    & 4    & 5    & 3    & 2    & 4    & 3    & 3    & 4    & 2    & 1    & 3    & 2    \\
3124    & 2    & 3    & 1    & 2    & 4    & 3    & 3    & 4    & 2    & 3    & 5    & 4    & 0    & 1    & 1    & 2    & 2    & 3    & 5    & 4    & 6    & 5    & 3    & 4    \\
3142    & 3    & 4    & 2    & 1    & 3    & 2    & 4    & 5    & 3    & 4    & 6    & 5    & 1    & 0    & 2    & 3    & 1    & 2    & 4    & 3    & 5    & 4    & 2    & 3    \\
3214    & 3    & 4    & 2    & 3    & 5    & 4    & 2    & 3    & 1    & 2    & 4    & 3    & 1    & 2    & 0    & 1    & 3    & 2    & 6    & 5    & 5    & 4    & 4    & 3    \\
3241    & 4    & 5    & 3    & 4    & 6    & 5    & 3    & 4    & 2    & 1    & 3    & 2    & 2    & 3    & 1    & 0    & 2    & 1    & 5    & 4    & 4    & 3    & 3    & 2    \\
3412    & 4    & 5    & 3    & 2    & 4    & 3    & 5    & 6    & 4    & 3    & 5    & 4    & 2    & 1    & 3    & 2    & 0    & 1    & 3    & 2    & 4    & 3    & 1    & 2    \\
3421    & 5    & 6    & 4    & 3    & 5    & 4    & 4    & 5    & 3    & 2    & 4    & 3    & 3    & 2    & 2    & 1    & 1    & 0    & 4    & 3    & 3    & 2    & 2    & 1    \\
4123    & 3    & 2    & 4    & 3    & 1    & 2    & 4    & 3    & 5    & 4    & 2    & 3    & 5    & 4    & 6    & 5    & 3    & 4    & 0    & 1    & 1    & 2    & 2    & 3    \\
4132    & 4    & 3    & 3    & 2    & 2    & 1    & 5    & 4    & 6    & 5    & 3    & 4    & 4    & 3    & 5    & 4    & 2    & 3    & 1    & 0    & 2    & 3    & 1    & 2    \\
4213    & 4    & 3    & 5    & 4    & 2    & 3    & 3    & 2    & 4    & 3    & 1    & 2    & 6    & 5    & 5    & 4    & 4    & 3    & 1    & 2    & 0    & 1    & 3    & 2    \\
4231    & 5    & 4    & 6    & 5    & 3    & 4    & 4    & 3    & 3    & 2    & 2    & 1    & 5    & 4    & 4    & 3    & 3    & 2    & 2    & 3    & 1    & 0    & 2    & 1    \\
4312    & 5    & 4    & 4    & 3    & 3    & 2    & 6    & 5    & 5    & 4    & 4    & 3    & 3    & 2    & 4    & 3    & 1    & 2    & 2    & 1    & 3    & 2    & 0    & 1    \\
4321    & 6    & 5    & 5    & 4    & 4    & 3    & 5    & 4    & 4    & 3    & 3    & 2    & 4    & 3    & 3    & 2    & 2    & 1    & 3    & 2    & 2    & 1    & 1    & 0   \\
\bottomrule
\end{tabular*}
\end{table}

\end{landscapeblock}
\clearpage




\end{document}